%% file: main.tex
\documentclass[runningheads]{llncs}
\usepackage[T1]{fontenc}
\usepackage{graphicx}

\input{setup/custom_packages}

\input{setup/custom_header}

\graphicspath{{./img/}{./img/svg/}}

\begin{document}
	
\label{beginningofdocument}
\newoutputstream{docstatus}
\openoutputfile{main.ds}{docstatus}
\addtostream{docstatus}{Preparing}
\closeoutputstream{docstatus}
	
\title{Quantitative Strategy Templates}
\author{Ashwani Anand \and
Satya Prakash Nayak\and
Ritam Raha \and
Irmak Sa\u{g}lam \and
Anne-Kathrin Schmuck
}
\institute{Max Planck Institute for Software Systems, Kaiserslautern, Germany\\
\email{\{ashwani,sanayak,rraha,isaglam,akschmuck\}@mpi-sws.org}
}
\authorrunning{A. Anand et al.}

\maketitle              %
\vspace*{-2em}
\begin{abstract}
This paper presents (permissive) \emph{Quantitative 
Strategy Templates} (QaSTels) to succinctly represent infinitely many winning strategies in two-player energy and mean-payoff games. This transfers the recently introduced concept of \emph{Permissive (qualitative) Strategy Templates} (PeSTels) for $\omega$-regular games to games with quantitative objectives. 
We provide the theoretical and algorithmic foundations of (i) QaSTel synthesis, and (ii) their (incremental) combination with PeSTels for games with mixed quantitative and qualitative objectives. Using a prototype implementation of our synthesis algorithms, we demonstrate empirically that QaSTels extend the advantageous properties of strategy templates over single winning strategies -- known from PeSTels -- to games with (additional) quantitative objectives. This includes (i) the enhanced robustness of strategies due to their runtime-adaptability, and (ii) the compositionality of templates w.r.t.\ incrementally arriving objectives. We use control-inspired examples to illustrate these superior properties of QaSTels for CPS design. 

\end{abstract}

\sloppy

\input{sections/introduction}

\input{sections/preliminaries}

\input{sections/Templates}

\input{sections/energyGames}

\input{sections/applications}

\input{sections/compositionality}

\input{sections/PeSTeLs}

\input{sections/experiments}

\label{beforebibliography}
\newoutputstream{pages}
\openoutputfile{main.pg}{pages}
\addtostream{pages}{\getpagerefnumber{beforebibliography}}
\closeoutputstream{pages}

 \bibliographystyle{splncs04}
 \bibliography{main}
 
 \newpage
 \input{sections/appendix}

\label{endofdocument}
\newoutputstream{pagestotal}
\openoutputfile{main.pgt}{pagestotal}
\addtostream{pagestotal}{\getpagerefnumber{endofdocument}}
\closeoutputstream{pagestotal}

\end{document}

%% file: setup/custom_packages.tex
\usepackage{mathtools,amsmath,amsfonts,amssymb}
\usepackage{thm-restate}
\usepackage{stmaryrd}
\usepackage{xspace}

\usepackage{algorithm}
\usepackage[noend]{algpseudocode}
\usepackage{changepage}
\usepackage{paralist}
\usepackage{thm-restate}
\usepackage{multirow}
\usepackage[capitalise]{cleveref}
\usepackage{graphicx}
\usepackage{subcaption}
\usepackage{caption}
\usepackage{longtable}
\usepackage{newfile}
\usepackage{refcount}
\usepackage{cite}

\crefname{item}{RQ}{RQ}

\usepackage[mathscr]{euscript}

\usepackage[
colorinlistoftodos, %
textwidth=\marginparwidth, %
textsize=scriptsize, %
]{todonotes}

\usepackage{mdframed}

\usepackage{pgf}
\usepackage{tikz}
\usetikzlibrary{shapes,shapes.geometric,fit,trees,decorations,arrows,automata,shadows,positioning,plotmarks,backgrounds,tikzmark}
\usetikzlibrary{calc,matrix,fit,petri,decorations.markings,decorations.pathmorphing,patterns,intersections,decorations.text}

\usepackage{graphicx}
\usepackage[utf8x]{inputenc}
\usepackage[american]{babel}
\usepackage[babel]{csquotes}

\usepackage{algorithm, algpseudocode}

\usepackage{amscd}
\usepackage{mathtools}
\usepackage{array}

\usepackage{xspace}
\usepackage{paralist}
\usepackage{xifthen}
\usepackage{url}

\usepackage{pgf}
\usepackage{tikz}

\usepackage{interval}
\usepackage{wrapfig}

%% file: setup/custom_header.tex
\usepackage{expl3}[2012-07-08]
\ExplSyntaxOn
\cs_if_exist:NTF \fpeval
{ %
}
{ %
	\cs_new_eq:NN \fpeval \fp_eval:n
}
\ExplSyntaxOff

\colorlet{darkgreen}{green!80!black}
\colorlet{darkred}{red!80!black}
\usetikzlibrary{arrows, automata, shapes}
\tikzset{auto, >= stealth}
\tikzset{every edge/.append style={thick, shorten >= 1pt}}
\tikzset{initial/.style={draw, thick, <-, shorten <=1pt}}
\tikzset{player0/.style = {draw, thick, shape=circle, minimum size=5mm}}
\tikzset{player1/.style = {draw, thick, shape=rectangle, minimum size=5mm}}
\tikzset{bplayer0/.style = {draw, thick, shape=ellipse, minimum size=5mm,text width=1.1cm}}
\tikzset{bplayer1/.style = {draw, thick, shape=rectangle, minimum size=5mm,text width=1.6cm}}

\usetikzlibrary{shapes.geometric, arrows}

\tikzstyle{startstop} = [rectangle, rounded corners, 
minimum width=0.1cm, 
minimum height=0.5cm,
text centered, 
draw=black, 
]

\tikzstyle{process} = [trapezium, 
trapezium stretches=true, %
trapezium left angle=70, 
trapezium right angle=110, 
minimum width=0.1cm, 
minimum height=0.5cm, 
draw=black, 
]

\tikzstyle{io} = [rectangle, 
minimum width=0.1cm, 
minimum height=0.5cm, 
text centered, 
draw=black, 
]

\tikzstyle{decision} = [diamond, 
aspect=3,
text centered, 
draw=black, 
]
\tikzstyle{arrow} = [thick,->,>=stealth]

\newcommand{\N}{\mathbb{N}}
\newcommand{\Q}{\mathbb{Q}}
\newcommand{\Z}{\mathbb{Z}}
\newcommand{\Ninf}{\mathbb{N}_{\infty}}
\newcommand{\Zinf}{\mathbb{Z}_{\infty}}
\newcommand{\bigO}{\mathcal{O}}

\newcommand{\abs}[1]{\left\lvert #1 \right\rvert}

\newcommand{\game}{\mathcal{G}}
\newcommand{\gamegraph}{\ensuremath{G}}
\newcommand{\play}{\ensuremath{\rho}}
\newcommand{\labellingFunc}{\mathbb{L}}
\newcommand{\weight}{w}
\newcommand{\average}{\texttt{avg}}

\newcommand{\parityLabel}{\labellingFunc_{P}}
\newcommand{\priority}{\parityLabel}

\newcommand{\spec}{\ensuremath{\Phi}}

\newcommand{\node}{V}

\newcommand{\p}[1]{\ensuremath{\text{Player}~#1}}
\newcommand{\pz}{\ensuremath{\p{0}}\xspace}
\newcommand{\po}{\ensuremath{\p{1}}}

\newcommand{\win}{\mathcal{W}}

\newcommand{\strat}{\ensuremath{\pi}}

\newcommand{\Strat}{\ensuremath{\Lambda}}

\newcommand{\minActivationEdges}{\ensuremath{\mathtt{minEdges}}}

\newcommand{\buchi}{\ifmmode B\ddot{u}chi \else B\"uchi \fi}
\newcommand{\cobuchi}{\ifmmode co\text{-}B\ddot{u}chi \else co-B\"uchi \fi}

\newcommand{\safetygame}{\mathit{Safety}}

\newcommand{\parity}{\ensuremath{\mathit{Parity}}}
\newcommand{\energy}{\ensuremath{\mathit{En}}}
\newcommand{\meanpayoff}{\ensuremath{\mathit{MP}}}

\newcommand{\computeTemp}{\textsc{computePeSTel}}

\newcommand{\cobuchiTemp}{\textsc{coB\"uchiTemp}\xspace}

\newcommand{\energyTemp}{\textsc{computeQaSTel}\xspace}

\newcommand{\template}{\ensuremath{\Gamma}}

\newcommand{\src}{\mathit{src}}
\newcommand{\livegroup}{H_\ell}

\newcommand{\livegroupSingleN}{H}
\newcommand{\colivegroup}{D}
\newcommand{\safegroup}{S}

\newcommand{\pref}{\textsf{pref}}

\makeatletter
\newsavebox{\@brx}
\newcommand{\llangle}[1][]{\savebox{\@brx}{\(\m@th{#1\langle}\)}%
	\mathopen{\copy\@brx\kern-0.5\wd\@brx\usebox{\@brx}}}
\newcommand{\rrangle}[1][]{\savebox{\@brx}{\(\m@th{#1\rangle}\)}%
	\mathclose{\copy\@brx\kern-0.5\wd\@brx\usebox{\@brx}}}

\makeatother

\makeatletter 
\newif\ifFIRST
\newif\ifSECOND
\let\LISTOP\relax
\newcommand{\List}[4][\;]{#3#1%
	\FIRSTtrue
	\@for\i:=#2\do{%
		\ifFIRST\LISTOP{\i}\FIRSTfalse\else,\LISTOP{\i}\fi%
	}%
	#1#4%
	\let\LISTOP\relax
}
\makeatother

\makeatletter

\newcommand{\tup}[1]{\left( #1\right)}

\newcommand{\VSet}[2][]{\let\LISTOP\val\List[#1]{#2}{\{}{\}}}

\newcommand{\VTuple}[2][]{\let\LISTOP\val\List[#1]{#2}{(}{)}}

\newcommand{\activation}{act}

\newcommand{\FindConflicts}{\ensuremath{\textsc{findConf}}\xspace}
\newcommand{\mixedTemp}{\ensuremath{\textsc{computeMiSTel}}\xspace}

\newcommand{\conflict}{\mathcal{C}}

\newcommand{\extract}{\textsc{ExtractStrat}}

\definecolor{colorblindgreen}{HTML}{004D40}
\definecolor{colorblindblue}{HTML}{1E88E5}
\definecolor{colorblindorange}{HTML}{FFC107}
\definecolor{colorblindred}{HTML}{D81B60}
\definecolor{colorblindbluealt}{HTML}{00BFFF}
\definecolor{colorblindorangealt}{HTML}{FF7F00}

\algdef{SE}[DOWHILE]{Do}{doWhile}{\algorithmicdo}[1]{\algorithmicwhile\ #1}%

\newcommand{\funcV}{\ensuremath{\mathcal{F}_V}}

\newcommand{\funcE}{\ensuremath{\mathcal{F}_{E}}}
\newcommand{\fixpoint}{\ensuremath{\mathbb{O}}}
\newcommand{\computeFixpoint}{\ensuremath{\textsc{FixPoint}}}
\newcommand{\fixpointV}{\ensuremath{\mathbb{O}_V}}

\newcommand{\fixpointE}{\ensuremath{\mathbb{O}_{E}}}

\intervalconfig{soft open fences, separator symbol={;}, colorize=\color{black}}
\newcommand{\closedopeninterval}[2]{\interval[open right]{#1}{#2}}

\renewcommand{\inf}{\mathtt{Inf}}
\newcommand{\plays}{\mathtt{plays}}

\newcommand{\Win}{\mathtt{Win}}

\newcommand{\aplays}{\mathtt{plays}}
\newcommand{\prefix}{H}
\newcommand{\energyTemplate}{\Pi}
\newcommand{\opt}{\mathtt{opt}}
\newcommand{\optE}{\mathtt{optE}}
\newcommand{\bound}[1]{\mathtt{B}_{#1}}
\newcommand{\finite}{\mathtt{f}}
\newcommand{\combine}{\textsc{Combine}}

\newcommand{\tool}{\texttt{QuanTemplate}\xspace}
\newcommand{\othertool}{\texttt{MPCoBuechi}\xspace}

\newcommand{\qastel}{QaSTel\xspace}

%% file: sections/introduction.tex
\section{Introduction}

Two player games on finite graphs provide a powerful abstraction for modeling the strategic interactions between reactive systems and their environment.  
In this context, game-based abstractions are 
often enriched with quantitative information to model aspects like energy consumption, cost 
minimization, or maintaining system performance thresholds under varying conditions. As a result, 
games with quantitative objectives, such as energy~\cite{energy} or mean-payoff~\cite{meanpayoff} have gained significant attention in recent years. These games have been applied to a wide range of CPS problems, such as energy management in electric vehicles~\cite{energy}, optimizing resource-constrained task management in autonomous robots~\cite{task,finitetask}, embedded 
systems~\cite{ChakrabartiAHS03}, and dynamic resource allocations~\cite{AVNI202042}.

In practical CPS applications, strategic control decisions (i.e., the moves of the controller 
player in the abstract game) are typically implemented via low-level actuators. %
For instance, a robot's strategic decision to move to a different room involves motion control integrated with LiDAR-based obstacle avoidance. However, due to unmodeled dynamics of the physical environment that become observable only at runtime, strategic adaptations may be necessary~\cite{Vazquez-Chanlatte21,GittisVF22}.
For example, if the robot detects that an entrance is obstructed by obstacles (e.g., humans), it should dynamically adjust its strategy and navigate through an alternative door instead.
Therefore, synthesized (high-level) control strategies must not only be correct-by-design, 
but also flexible enough to accommodate runtime adaptations. This control-inspired property of strategies has 
recently been formalized via permissive strategy templates 
(PeSTels)~\cite{Anandetal_SynthesizingPermissiveWinning_2023}, which are similar to classical 
strategies but contain a vast set of relevant strategies in a succinct and simple data 
structure. Intuitively, PeSTels \emph{localize} required progress towards $\omega$-regular 
objectives by classifying outgoing edges of a control-player vertex as unsafe, co-live and 
live -- indicating edges to be taken never, finitely often and infinitely often, respectively, 
in case the source vertex of the edge is visited infinitely often.
 
Inspired by PeSTels and driven by the need to capture quantitative objectives in CPS design, this paper introduces \emph{Quantitative Strategy Templates} (QaSTels) for games with energy and mean-payoff objectives. n the context of the previously discussed robot example, such games model scenarios where a robot with limited battery
must make informed re-routing decisions at runtime, ensuring that its remaining energy suffices for the required tasks.
 Similar to PeSTels, QaSTels localize necessary information about the \emph{future} of the game. In contrast to PeSTels, which localize \emph{liveness} requirements induced by a \emph{qualitative} objective, QaSTels consider \emph{quantitative} objectives and thereby localize the required energy loss and gain through local edge annotations. %
Knowing the current energy level at runtime, the control player can select from all edges that remain feasible given the available energy. This contrasts with standard game-solving approaches, which typically store only a single (optimal) action per node.

Concretely, our contributions are as follows:
\begin{inparaenum}[(i)]
    \item We formalize QaSTels for energy and mean-payoff 
    objectives, and present algorithms to extract winning strategies from them.
    \item We introduce an edge-based value iteration algorithm %
    to compute winning %
    QaSTels and show that QaSTels are \emph{permissive}, i.e., they capture all winning strategies for energy objectives and all finite-memory winning strategies for 
    mean-payoff objectives. %
    \item We combine QaSTels with a bounded version of PeSTels, and propose an \emph{efficient incremental algorithm} for 
    updating templates and strategies 
    under newly arriving qualitative and quantitative objectives.
   \item We highlight the advantages of strategy templates for games with quantitative objectives, and the combination of quantitative and qualitative objectives, via extensive experiments on benchmarks derived from the SYNTCOMP benchmark suit.
\end{inparaenum}
Detailed proofs for all claims are provided in the appendix.

\smallskip
\noindent\textbf{Related Work.}
The computation of permissive strategies has received significant attention over the past decade, particularly for qualitative objectives~\cite{bernet2002:permissiveStrategies,neider2014:mullergamestosafety, bouyer2011:measuringpermissiveness, Klein2015:mostGeneralController}. A key development in this area is the introduction of permissive strategy templates (PeSTels) by Anand et al.~\cite{Anandetal_SynthesizingPermissiveWinning_2023,Anandetal_ComputingAdequatelyPermissive_2023}, which capture a strictly broader class of winning strategies while maintaining the same worst-case computational complexity as standard game-solving techniques. This paper extends the idea behind PeSTels to quantitative games.%

When only `classical' synthesis algorithms are available, achieving the adaptability of strategies that motivate PeSTels requires recomputing a new 
strategy from scratch whenever moves become unavailable at runtime or additional objectives arise.
For the objectives considered in this paper, this entails using `classical' synthesis algorithms for energy objectives~\cite{energy}, mean-payoff objectives~\cite{meanpayoff,fastermeanpayoffgames}, multi mean-payoff objectives~\cite{multimeanpayoff}, and mean-payoff co-Büchi objectives~\cite{ChatterjeeHJ05,ChatterjeeHS17}. However, since these approaches recompute strategies from scratch at each iteration, they are computationally expensive. Our benchmark experiments demonstrate that QaSTel-based adaptations offer a more efficient alternative for the applications considered. %

%% file: sections/preliminaries.tex
\section{Preliminaries}\label{sec:preliminaries}
In this section, we introduce the basic notations used throughout the paper. We denote $\Z$ as 
the set of integers, $\Q$ as the set of rational numbers, $\N$ as the set of natural numbers 
including 0, and $\N_{>0}$ as the set of positive integers.
Let $ \Ninf = \N\cup \{\infty\} $ and $ \Zinf = \Z\cup \{\infty, -\infty\} $.
The interval $ \closedopeninterval{a}{b} $ represents the set $ \{a, a+1, \cdots , b\} $. 

\subsection{Two-Player Games}	
A two-player game graph is a pair $G= (V, E)$, where $V= V_0 \uplus V_1$ is a 
finite set of nodes, $E \subseteq V \times V$ is a set of edges.
The nodes are partitioned into 
two sets, $V_0$ and $V_1$, where $V_i$ represents the set of nodes controlled by Player $i$ for $i \in 
\{0,1\}$. 
Further, we write $E_i$ to denote the set of edges originating from nodes in $V_i$, i.e., $E_i = E \cap (V_i \times V)$.
Given a node $v$, we write $E(v)$ to denote the set $\{e \in E \mid e = (v, v') \text{ 
for } v' \in V\}$ of all outgoing edges from $v$.

\smallskip
\noindent\textbf{Value Functions.} 
For a set of nodes $V$, and a set of edges $E$ let $ \funcV $ denote $ \{ f \mid f: V\rightarrow \Ninf\} $, and $ \funcE $ denote $ \{ f \mid f: E\rightarrow \Ninf\} $.

\smallskip
\noindent\textbf{Plays.}
A \emph{play} $\play = v_0 v_1 \ldots \in V^\omega$ on $G$ is an 
infinite sequence of nodes starting from $v_0$ such that, $(v_i, v_{i+1}) 
\in E$ for each $i$. We denote the $i^{th}$ node $ v_i $ of $\play$ as $\play[i]$ and 
use the notations $\play[0 ;i] = v_0 \ldots v_i$, $\play[i;j] = v_i \ldots v_j$, and $ \play[i;\infty] = 
v_i \ldots $ to denote a \emph{prefix},  \emph{infix}, and \emph{suffix} of $\play$, respectively. 
We write $v\in\play$ (resp.\ $e\in\play$) to denote that the node $v$ (resp. the edge $e$) appears in $\play$.
Furthermore, we write $v\in\inf(\play)$ (resp.\ $e\in\inf(\play)$) to denote that node $v$ (resp.\ edge $e$) appears infinitely often in the play $\play$.%
We denote by $\plays(G)$ the set of all plays on $G$, by $\plays(G,v)$ denote the set of all plays starting from node $v$.

\smallskip
\noindent\textbf{Strategies.} 
A \emph{strategy} $\strat$ for $\p{i}$, where $i \in \{0,1\}$ (or, a \emph{
$\p{i}$-strategy}) is a function  $\strat: V^* \cdot V_i \mapsto E$ such that for all $H \cdot v 
\in V^* \cdot V_i$, we have $\strat(H \cdot v) \in E(v)$. %
A play $\play = v_0 v_1 \ldots$ is called a 
$\strat$-play if it follows $\strat$, i.e., for all $j \in \N$, whenever $v_j \in V_i$ it holds that $\strat(v_0 \ldots v_j) = (v_j,v_{j+1})$. Given a strategy 
$\strat$, we write $\plays_{\strat}(G,v)$ to denote the set of all $\strat$-plays starting 
from node $v$ and $\plays_{\strat}(G)$ to denote the set of all $\strat$-plays in $G$. 
For an edge $e$, we write $\plays_{\strat}(G,e)$ to denote the set of plays that start with $e$ and follows $\strat$, i.e., a play $\play\in \plays_{\strat}(G,e)$ iff 
$(\play[0],\play[1]) = e$ and $\play[1; \infty] \in \plays_{\strat}(G,\play[1])$.

Let $M$ be a \emph{memory} set. A $ \p{i} $-strategy $\strat$ with memory $M$ is represented 
as a tuple $(M, m_0, \alpha, \beta)$, where $m_0 \in M$ is the initial memory value, $\alpha : M 
\times V \to M$ is the memory update function, and $\beta : M \times V_i \to V$ is the state transition 
function. Intuitively, if the current node is a $ \p{i} $ node $v$ and the current memory 
value is $m$, the strategy $\strat$ selects the next node $v' = \beta(m, v)$ and updates the memory to
$\alpha(m, v)$. If $M$ is finite, we call $\strat$ a \emph{finite-memory strategy}; otherwise, 
it is an \emph{infinite-memory strategy}. 
Formally, given a history $H\cdot v\in V^*\cdot V_i$, the strategy is defined as $\strat(H\cdot v) = 
\beta(\hat{\alpha}(m_0, H), v)$, where $\hat{\alpha}$ is the canonical extension of $\alpha$ to sequences of nodes.
A strategy is called \emph{memoryless} or \emph{positional} if $|M| = 1$. For a memoryless 
strategy $\strat$, it holds that $\strat(H_1 \cdot v) = \strat(H_2 \cdot v)$ for every history 
$H_1,H_2 \in V^*$. For convenience, we write
$\strat(v)$ instead of $\strat(H\cdot v)$ for such strategies. 

For a game graph $G = (V,E)$ with a finite-memory strategy $\strat = (M, m_0, \alpha, \beta)$, we denote by $G_\strat = (V',E')$ the product of $G$ and $\strat$, that is, $V_0' = V_0 \times M$, $V_1' = V_1\times M$, and $E' = \{((v,m),(v',m')) \mid (v,v') \in E, m' = \alpha(m,v)\}$. 
With slight abuse of terminology, we say that a state $v$ is \emph{reachable from $q$ in $G_\strat$} if 
there exists a tuple $(v,m')$ reachable from $(q, m_0)$ in $G_\strat$. Similarly, we say that a sequence $v_0 v_1 
\ldots$ is a play in $G_\strat$ if there exists a corresponding play $(v_0, m_0)(v_1, m_1) \ldots$ in $G_\strat$.

\smallskip
\noindent\textbf{Games and Objectives.} 
A \emph{game} is a tuple $(G, \varphi)$, where $G$ is a game graph and $\varphi \subseteq V^\omega$ is an \emph{objective} for $\pz$. 
A play $\play$ is considered \emph{winning} if $\play \in \varphi$.
A $\p{0}$ strategy $\strat$ is \emph{winning from a node $v$}, if all $\strat$-plays starting from $v$ are winning.
Similarly, $\strat$ is winning from $V'\subseteq V$ if it is winning from all nodes in $V'$.
We define the \emph{winning region} $\Win(G, \varphi)$ as
the set of nodes from which $\p{0}$ has a winning strategy in $(G, \varphi)$.
A $\p{0}$ strategy is \emph{winning} if it is winning from $\Win(G, \varphi)$.

We define a \emph{weight function} $\weight : E \to [-W;W]$ for some $W \in \N_{>0}$, which assigns an integer weight to each edge in $G$. This function extends naturally to finite infixes of plays, i.e., $\weight(v_0v_1\ldots v_k) = \sum_{i=0}^{k-1} \weight(v_i,v_{i+1})$. Furthermore, we define the average weight of a finite prefix $v_0v_1\ldots v_k$ as $\average(v_0v_1\ldots v_k) = \frac{1}{k} \sum_{i=0}^{k-1} \weight(v_i,v_{i+1})$.
With this, we consider the following objectives in games:\\
\begin{inparaitem}[$\triangleright$]
	\item \emph{(Quantitative) Energy Objectives}. 
	Given a weight function $\weight$ and an initial credit $c \in\N$, the \emph{energy objective} is defined as $\energy_c(\weight) = \{\play \in V^\omega \mid c+ \weight(\play[0;i]) \geq 0,\ \forall i \in \N\}$.
	Intuitively, the energy objective ensures that the total weight (`energy level') remains non-negative along a play.\\
	\item \emph{(Quantitative) Mean-Payoff Objectives.} 
	Given a weight function $\weight$, the \emph{mean-payoff objective} is defined as\footnote{We note that mean-payoff objectives can also be defined via the \emph{limit-inferior} function i.e., $\{\play \in V^\omega \mid \liminf_{n\to\infty} \average(\play[0;n]) \geq 0\}$. However, it has been shown that games with either definition are equivalent~\cite[Corollary 8]{gogbook}.} $\meanpayoff(\weight) = \{\play \in V^\omega \mid \limsup_{n\to\infty} \average(\play[0;n]) \geq 0\}$. 
	Intuitively, the mean-payoff objective ensures that the limit average weight of a play is non-negative.	\\
	\item \emph{(Qualitative) Parity Objectives}. Given a \emph{priority labeling} 
	$\priority: V \to [0;d]$ for some $d \in \N_{>0}$, which assigns a priority to each node in $G$, the \emph{parity objective} is defined as
	$\parity(\priority) = \{\play \in V^\omega \mid \max_{v\in\inf(\play)} \priority(v) \text{ is even}\}$. Intuitively, the parity objective ensures that the highest priority seen infinitely often along a play is even.
	
\end{inparaitem}
We refer to a game with a mean-payoff, energy, or parity objective as a \emph{mean-payoff game}, \emph{energy game}, and \emph{parity game}, respectively. A game is called \emph{mixed} if it is equipped with a conjunction of  quantitative and qualitative objectives. Further, we call $G$ weighted, if it is annotated with a weight function $\weight$, denoted by $G_\weight$. Consequently, energy and mean-payoff are referred to as weighted games.

\smallskip
\noindent\textbf{Fixed and Unknown Initial Credit Problem.}
We consider the following game variants for energy objectives.
\begin{inparaenum}[\bfseries (1)]
	\item Given an initial credit $c$, the \emph{energy game with fixed initial credit $c$} is defined as the game $(G, \energy_c(\weight))$.
	\item A game $(G,\energy(\weight))$ with \emph{unknown initial credit} asks $\pz$ to ensure the objective $\energy_c(\weight)$ for some finite initial credit $c$.
\end{inparaenum}

In an energy game $(G = (V,E),\energy(\weight))$, there exists an optimal initial credit $\opt(v)\in\Zinf$ for each node $v$, where $\opt(v)$ is the minimal value (in $\Ninf$) such that for every initial credit $c \geq \opt(v)$, there exists a winning strategy from $v$ in the game $(G, \energy_c(\weight))$. We use $\opt\in\funcV$ to denote this \emph{optimal value function} which assigns the optimal initial credit to each node in the game graph.
It is well-known that the optimal initial credit is upper bounded by $c^* = W\cdot |V|$, where $W$ is the maximum weight in the weight function $\weight$.
Hence, for any initial credit $c\geq c^*$, the winning region for energy game $(G,\energy(\weight))$ with unknown initial credit is the same as the winning region for energy game $(G,\energy_{c}(\weight))$. Moreover, every winning strategy in $(G,\energy_{c}(\weight))$ is also winning in $(G,\energy(\weight))$.

%% file: sections/Templates.tex
\section{Quantitative Strategy Templates (QaSTels)}\label{sec:Templates}

In this section, we first define a quantitative strategy template (QaSTel) for weighted games and show how it can be used to represent the set of strategies in a weighted game. We then define winning and maximally permissive QaSTels. 

\begin{definition}[Quantitative Strategy Template (QaSTel)]
	Given a game graph $G = (V,E)$, a \emph{QaSTel} %
	for $ \pz $ is a function $ 
	\energyTemplate:V_0\times \Ninf \rightarrow 2^E $ that maps a $ \pz $ node $u$ and the 
	current credit $c$ to a subset of outgoing 
	edges of $u$ in $G$ that are \emph{activated} by $c$ s.t.\ $\energyTemplate(u, c) \subseteq \energyTemplate(u, c')$ for all $c'\geq c$.
\end{definition}
We also use a QaSTel $\energyTemplate$ as a function $V_0\times \Zinf \rightarrow 2^E$ by extending it to negative credits as follows: $ \energyTemplate(v, c) = \emptyset $ for all $ c < 0 $.
If $ \energyTemplate(u, i) = \energyTemplate(u, i+1) = \cdots = \energyTemplate(u, j) = E' $, then for notational simplicity, we will write $ \energyTemplate(u, [i; j]) = E' $. 
This naturally defines the \emph{activation function} for an edge $e$, denoted by $\activation_\energyTemplate(e)$, as the smallest value $k$ such that $e\in \energyTemplate(u, [k; \infty])$.

Given a weighted game $G_\weight$ with a QaSTel $\energyTemplate$ and a weight $c \in\N$, a play $ \play = v_0v_1\cdots $ is said to be a $(\energyTemplate,c)$-play if 
there exists a $k\in\Ninf$ such that for all $i\in [0;k]$ 
with $ v_i\in V_0 $, $ (v_i,v_{i+1})\in 
\energyTemplate(v_i, c + \weight(\play[0;i])) $ and if 
$k\neq \infty$, then %
whenever $v_{k+1}\in V_0$ it holds that
$\energyTemplate(v_{k+1}, c + \weight(\play[0;k+1])) = 
\emptyset$.
Intuitively, either the play only uses the active edges from the QaSTel forever, or it reaches a node where no edge is active in the QaSTel and then the play continues with arbitrary edges.
We collect all $(\energyTemplate,c)$-plays in $G$ from a node $ v $ in the set $ \aplays_{\energyTemplate}(G,c,v) $.
Similarly, we write $ \aplays_{\energyTemplate}(G,c) $ to denote the set of all $(\energyTemplate,c)$-plays in $G$.

QaSTels define a set of ($\pz$) 
strategies in a weighted game which \emph{follow it}, as formalized next. 
\begin{definition}
Given an energy game 
	$(G,\energy_c(\weight))$ with initial credit 
	$c\in\N$, 
	a strategy $\strat$ is said to \emph{follow} a QaSTel $\energyTemplate$, denoted by $(G,\strat) 
	\vDash_{c}\energyTemplate$ (or simply $\strat \vDash_{c}\energyTemplate$ when $G$ 
is clear from the context),
if $ \plays_\strat(G)\subseteq \aplays_{\energyTemplate}(G,c)$.
	Similarly, for a mean-payoff game $(G,\meanpayoff(\weight))$, a strategy $\strat$ is said 
	to \emph{follow} a QaSTel $\energyTemplate$ if $(G,\strat) 
	\vDash_{c}\energyTemplate$ for some $c \geq W \cdot |V|$.
\end{definition}
For mean-payoff games, the previous definition chooses the initial credit $c$ to be at least the upper bound on the optimal credit, i.e., $c \geq W \cdot |V|$.
This is motivated by the fact that mean-payoff games are equivalent to energy games with unknown initial credit~\cite{fastermeanpayoffgames}. Therefore, winning strategies of mean-payoff games can be captured by winning strategies of energy games with credit above the upper bound on the optimal credit.

In the upcoming definition, we define a \emph{winning} QaSTel.
\begin{definition}[Winning QaSTel]
	Given a weighted game, a QaSTel $\energyTemplate$ is said to be \emph{winning} from a node $v$ (resp.\ a set $V'$ of nodes) if every strategy following $\energyTemplate$ is also winning from $v$ (resp.\ $V'$).
	Furthermore, a QaSTel $\energyTemplate$ is said to be \emph{winning} if it is winning from the winning region.
\end{definition}
A winning QaSTel is \emph{maximally permissive} if it includes all winning strategies. %
\begin{definition}[Maximal Permissiveness]
	Given a weighted game, a QaSTel $ \energyTemplate $ is said to be \emph{maximally permissive} if every winning strategy follows $\energyTemplate$.
	Furthermore, a QaSTel $ \energyTemplate $ is said to be \emph{$\finite$-maximally permissive} if every winning strategy with finite memory follows $\energyTemplate$.
\end{definition}

Given the simple and local structure of QaSTels, one can easily extract a positional strategy for 
$\pz$ following the QaSTel by 
picking the edge with the smallest activation value at every node. This clearly results in a 
winning strategy if the QaSTel is 
winning. %
\begin{restatable}{proposition}{restatePextractStrategy}\label{prop:extractStrategy}
	Given a weighted game graph $G_w$ s.t.\ $G = (V,E)$ with a QaSTel $\energyTemplate$, a positional strategy $\strat$ following $\energyTemplate$ can be extracted in time $\bigO(\abs{E})$.
	Let $\extract(G,\weight,\energyTemplate)$ be the procedure extracting this strategy.
\end{restatable}

\begin{proposition}
	Given a weighted game with game graph $G_\weight$ and a winning QaSTel $\energyTemplate$, the strategy $\extract(G,\weight,\energyTemplate)$ is winning. 
\end{proposition}

\section{Synthesizing QaSTels}
We now discuss the synthesis of QaSTels over weighted games. 
As QaSTels are defined on edges, we first introduce an edge-optimal value function and an edge-based value iteration algorithm for weighted games. Then, we show how to extract \emph{optimal} QaSTels from the edge-optimal value function.
Finally, we show that optimal QaSTels are winning and permissive.

\smallskip\noindent\textbf
{Edge-based Value Iteration.}
It is known that both energy games and mean-payoff games (with threshold $0$ as considered in 
objective $\meanpayoff(\weight)$) have the same value iteration 
algorithm~\cite{fastermeanpayoffgames}.
To simplify the presentation we therefore restrict the discussion to energy games.

Recall that in energy games, $\opt(v)$ is the minimal credit required to win the game from node $v$.
We extend this notion to $\optE(e)$ such that $\optE(e)$ is the minimal credit required to take the edge $e$ and win the energy game from the source node of $e$.
Formally, for some edge $e = (u,v)$, $\optE(e)$ is the minimal value (in $\Ninf$) such that for any initial credit $c\geq \optE(e)$, there exists a $\pz$ strategy $\strat$ with $\plays_{\strat}(G,e) \subseteq \energy_c(\weight)$.
To compute $\optE$, we extend the standard value iteration algorithm to an edge-based value iteration algorithm. %

Given a weighted game graph $G_\weight$, %
the standard value iteration algorithm computes the least fixed point of the operator $ \fixpointV: \funcV\rightarrow\funcV $ defined as:

\begin{equation}
	\fixpointV(\mu)(u) = 
			\begin{cases}
				\min\{(\mu(v) \ominus \weight(e)): e = (u, v)\in E  \},\text{ if } u\in V_0 \\
				\max\{(\mu(v) \ominus \weight(e)): e = (u, v)\in E  \},\text{ if } u\in V_1				
			\end{cases}
\end{equation}
where $ l \ominus w= \max(l-w,0) $. 
This fixed-point computation is initialized with an initial function $ \mu_{in}: V \rightarrow\Z $, and each value is upper bounded by $\abs{V}\cdot W$, i.e., once a value reaches $\abs{V}\cdot W +1$, we replace it by $\infty$.
Let us denote this procedure of fixed-point computation of an operator $\fixpoint$ starting from an initial function $\mu_{in}$ as $\computeFixpoint(G,\weight,\fixpoint,\mu_{in})$.
Then, the optimal value function $\opt$ can be obtained by $\computeFixpoint(G,\weight,\fixpointV,\mu_{0})$ where $\mu_{0}(v) = 0$ for all $v\in V$.

To compute the edge-optimal value function $\optE$, we modify the value iteration algorithm by extending the operator $\fixpointV$ from functions over vertices to functions over edges. Hence, we define $ \fixpointE: \funcE\rightarrow\funcE $ for an edge $e = (u,v)$ as:
\begin{equation}
	\fixpointE(\mu)(e) = 
	\begin{cases}
		\min\{\mu(e') \ominus w(e): e'\in E(v) \},\text{ if } v\in V_0 \\
		\max\{\mu(e')\ominus w(e): e'\in E(v) \},\text{ if } v\in V_1.
	\end{cases}
\end{equation}

\begin{remark}\label{rem:optVoptE}
It is not hard to see that the operators $ \fixpointE $ and $ \fixpointV $ are closely related.
In particular, if $\mu_i^V\in\funcV$ and $\mu_i^E\in\funcE$ are the corresponding value 
functions obtained in the $i$-th 
iteration of $\fixpointV$ and $\fixpointE$ respectively, then $\mu_i^E(e) = 
\mu_i^V(v)\ominus\weight(e)$ for every edge 
$e=(u,v)$.
	This leads to a similar relation between $\opt$ and $\optE$, and hence,
	one can also obtain the optimal QaSTel using the standard node-based value iteration algorithm.
	However, our choice of presenting the edge-based approach 
	allows us to explain our idea better, at no additional cost. %
\end{remark}

With \cref{rem:optVoptE}, the following theorem directly follows from the properties 
of the standard value iteration algorithm.
\begin{theorem}\label{thm:edgeOptimal}
	Given a game graph $G = (V,E)$ and weight function $\weight : E \rightarrow [-W,W]$, the fixed-point $\computeFixpoint(G,\weight,\fixpointE,\mu_{0})$ is the edge-optimal value function $\optE$ and can be computed in time $\bigO(\abs{V}\cdot\abs{E}\cdot W)$.
\end{theorem}

	Given a weighted game graph $G_\weight$, the winning region $\win$ for both mean-payoff and energy games with unknown initial credit can be extracted from the edge-optimal value function $\mu = \optE$ as 
	\begin{subequations}\label{equ:optimalToWinning}
	 \begin{equation}
	  \win := \{v\in V_0 \mid \exists e\in E(v).\ \mu(e) \neq \infty\} \cup \{v\in V_1 \mid \forall e\in E(v).\ \mu(e) \neq \infty\}.
	 \end{equation}
	Furthermore, for energy games with initial credit $c$ we obtain the winning region
	\begin{equation}
	 \win_c := \{v\in V_0 \mid \exists e\in E(v).\ \mu(e) \leq c\} \cup \{v\in V_1 \mid \forall 
	 e\in E(v).\ \mu(e) \leq c\}.
	\end{equation}
	\end{subequations}

\smallskip\noindent\textbf
{QaSTel Extraction.}
Given an edge function $\mu\in\funcE$, we can extract a QaSTel $\energyTemplate$ from $\mu$ as follows.
For every node $u\in V$ and credit $k\in \Ninf$, $\energyTemplate(u,k)$ defines the set of edges that can be taken from $u$ with credit $k$, i.e.,
\begin{equation}\label{equ:eTemp}
 \energyTemplate(u,k) := \{e\in E(u) \mid k\geq \mu(e)\}.
\end{equation}
Intuitively, in an energy game, the QaSTel in \eqref{equ:eTemp} allows taking an edge $e$ whenever its feasible w.r.t. edge function $\mu$, i.e., the current energy is more than the edge value $\mu(e)$.
We call the QaSTel in \eqref{equ:eTemp} \emph{optimal} for the weighted game graph $G_\weight$ if $\mu$ is the edge-optimal value function $\optE$.
Given an initial edge function $\mu_{in}$, we write $\energyTemp(\gamegraph,\weight,\mu_{in})$ to denote the procedure that computes the fixed-point $\mu = \computeFixpoint(G,\weight,\fixpointE,\mu_{in})$ and returns the corresponding winning region (as in \eqref{equ:optimalToWinning}) and the corresponding QaSTel obtained from the $\mu$ (as in \eqref{equ:eTemp}).
This means, the optimal QaSTel can be obtained by the procedure $\energyTemp(\gamegraph,\weight,\mu_{0})$ (where $\mu_{0}$ is the zero function on edges).
An example of the computation of the optimal QaSTel is shown in \Cref{fig:Example}.

\begin{figure}[t]
    \centering
    \begin{minipage}{0.45\textwidth}
        \centering
        \begin{tabular}{|c|cccccccc|}
			\hline
			\multicolumn{9}{|c|}{~Edge-based Value Iteration~~}\\
            \hline
            - & $e_1$ & $e_2$ & $e_3$ & $e_4$ & $e_5$ & $e_6$  & $e_7$     & $e_8$     \\
            \hline
			$\mu_0$  &  0     & 0     & 0     & 0     & 0     & 0      & 0         & 0     \\\cline{1-9} 
			$\mu_1$  &  0     & 2     & 5     & 2     & 0     & 0      & 0         & 1     \\\cline{1-9} 
			$\mu_2$  &  0     & 2     & 5     & 2     & 0     & 1      & 0         & 2     \\\cline{1-9} 
			$\mu_3$  &  0     & 2     & 5     & 2     & 0     & 2      & 0         & 3     \\\cline{1-9} 
			$\mu_4$  &  0     & 2     & 5     & 2     & 0     & 3      & 0         & 4     \\\cline{1-9} 
				     &  	  &       &       &       &       &        &           & \vdots\\\cline{1-9} 
			$\mu_{15}$  &  0     & 2     & 5     & 2     & 0     & 14      & 0         & $ 15 $    \\\cline{1-9} 
			$\mu_{16}$  &  0     & 2     & 5     & 2     & 0     & 15      & 0         & $\infty $    \\\cline{1-9} 
			$\mu_{17}$  &  0     & 2     & 5     & 2     & 0     & $ \infty $ & 0         & $ \infty $  \\
            \hline 
        \end{tabular}%
    \end{minipage}
	\hfill
	\begin{minipage}{0.45\textwidth}
		\begin{tikzpicture}
			\hspace{-1cm}
			\node[player0] (0) at (-2.5,0) {$a$};
			\node[player0] (1) at (0, 0) {$b$};
			\node[player1] (2) at (2.5, 0) {$c$};
			
			\path[->] (0) edge[loop above] node {$ e_1 = +1 $} (0) 
						  edge[loop below] node {$ e_2 = -2 $} (0) 
						  edge[bend left = 20] node {$ e_3 = -5 $} (1);
			\path[->] (1) edge[bend left = 20] node {$ e_4 = -2 $} (0) 
						  edge[loop above] node {$ e_5 = +1 $} (1) 
						  edge[bend left = 20] node {$ e_6 = 0 $} (2);
			\path[->] (2) edge[bend left = 20] node {$ e_7 = 0 $} (1)
						  edge[loop above] node {$ e_8 = -1 $} (2);
	\end{tikzpicture}
        \centering
        \begin{eqnarray*}
			\vspace{-1cm}
			(a, \closedopeninterval{0}{2}) &\mapsto& \{e_1\}\\
			(a, \closedopeninterval{2}{5}) &\mapsto& \{e_1, e_2\}\\
			(a, \closedopeninterval{5}{\infty}) &\mapsto& \{e_1, e_2, e_3\}\\
			(b, \closedopeninterval{0}{2}) &\mapsto& \{e_5\}\\
			(b, \closedopeninterval{2}{\infty}) &\mapsto& \{e_5,e_4\}\\
		\end{eqnarray*}
		
    \end{minipage}%
	\caption{Example of an energy game (right top) with the computation for the edge-based value iteration (left) and the optimal QaSTel (right bottom).}
	\label{tab:complex example computation}\label{fig:Example}
\end{figure}

%% file: sections/energyGames.tex
\smallskip\noindent\textbf
{Winning and Maximally Permissive QaSTels.}\label{sec:winningTemplates}
We now show that optimal QaSTels are winning for weighted games, and ($\finite$-)maximally 
permissive for (mean-payoff) energy games.
As a play defined by an optimal QaSTel only takes an edge if the credit is higher than its edge-optimal value, it is winning in the energy game.
Furthermore, the equivalence of energy and mean-payoff games gives the following result.
\begin{restatable}{theorem}{restateWinning}\label{thm:Winning}
	Given a weighted game graph $G_\weight$, the optimal QaSTel $ \energyTemplate $ is winning in both the mean-payoff game $(G,\meanpayoff(\weight))$ and the energy game $(G,\energy_c(\weight))$ for every initial credit $c\in \N$.
\end{restatable}

As an optimal QaSTel allows every edge ensuring positive energy w.r.t. the current credit, it is maximally permissive in an energy game.

\begin{restatable}{theorem}{restateEnergyMaximal}\label{thm:energyGamePermissive}
	Given a weighted game graph $G_\weight$, the optimal QaSTel $ \energyTemplate $ is maximally permissive in the energy game $(G,\energy_c(\weight))$ for every $c\in \N$.
\end{restatable}

Unlike in energy games, the optimal QaSTels are \emph{not} maximally permissive in mean-payoff games. However, it can capture all winning strategies with finite memory, i.e., it is $\finite$-maximally permissive.
To show this, we use the following property of finite memory winning strategies in mean-payoff 
games.%

\begin{restatable}{lemma}{restateLMeanpayoffFiniteMemory}\label{lemma:meanpayoffFiniteMemory}
	Let $(G,\meanpayoff(\weight))$ be a mean-payoff game with finite memory winning strategy $\strat$. Then there exists a weight bound $\bound{\strat}\in \N$ such that for every 
	$\strat$-play $\play$ from a node $v\in\win(G,\meanpayoff(\weight))$, it holds 
	that $\weight(\play[0;i])\geq -\bound{\strat}$ for all $i\in \N$.
\end{restatable}

With the above lemma, one can see that every winning strategy $\strat$ in the mean-payoff game is a winning strategy in the energy game with initial credit $c = \max\{\bound{\strat},W\cdot\abs{V}\}$.
Combining this with \cref{thm:energyGamePermissive}, we get the following result.

\begin{restatable}{theorem}{restateTmeanpayoffGamePermissive}\label{thm:meanpayoffGamePermissive}
	Given a weighted game graph $G_\weight$, the optimal QaSTel $ \energyTemplate $ is $\finite$-maximally permissive in the mean-payoff game $(G,\meanpayoff(\weight))$.
\end{restatable}

This shows that a winning and permissive QaSTel can be obtained by the procedure $\energyTemp(G,\weight,\mu_0)$, giving us the following result.
\begin{corollary}
	Given a weighted game graph $G_\weight$, a winning and maximally permissive QaSTel for the energy game $(G,\energy_c(\weight))$ can be computed in time $\bigO(\abs{V}\cdot\abs{E}\cdot W)$.
	Similarly, a winning and $\finite$-maximally permissive QaSTel for the mean-payoff game $(G,\meanpayoff(\weight))$ can be computed in time $\bigO(\abs{V}\cdot\abs{E}\cdot W)$.
\end{corollary}

%% file: sections/applications.tex
\section{Applications of QaSTels}
As discussed in the introduction, our study of QaSTels is inspired by the advantages their 
qualitative counterparts -- permissive strategy templates (PeSTels) for Parity games introduced 
in \cite{Anandetal_SynthesizingPermissiveWinning_2023} -- posses over classical strategies in 
control-inspired applications. In particular, PeSTels allow (i) to adapt winning strategies at 
runtime \cite{Anandetal_SynthesizingPermissiveWinning_2023,ContextTriggerdABCD_2023}, and 
(ii) to compose different templates into new ones leading to novel iterative and 
compositional synthesis techniques \cite{tacas24-equilibria,Anandetal_ContractBasedDistributedLogical_2024,AnandNS24}.
This section investigates whether QaSTels possess similar 
runtime adaptability (\cref{sec:LocalAdapt}) and compositional (\cref{sec:nocomp}) properties.

\subsection{Dynamic Strategy Extraction from QaSTels}\label{sec:LocalAdapt}
We first consider scenarios where the runtime operation of the controlled system, e.g.\ a robot, is 
supplied with local preferences over moves that can only be determined at runtime. As an example, 
consider a mobile robot in a smart factory operating in the presence of (non-modeled) human operators. 
Here, the robot might be equipped with a perception module which predicts the probability of the 
successful completion of an action (e.g.\ reaching a certain work station) within these (dynamic) 
obstacles. In this case, the logical control strategy can choose the activated move from the QaSTel 
with the highest success probability. More generally, given a weighted game $(G,\varphi)$ with optimal 
QaSTel $\energyTemplate$ and a dynamic preference function $\pref_t: E_0 \rightarrow [0,1]$ %
for every $t\in \mathbb{N}_0$, we define the $\pz$ strategy $\strat_0$ \emph{online} 
(after obtaining $\pref_t$ in time step $t$) s.t.\
\begin{equation}\label{equ:strat_adapt}
 \strat_0(v_0\hdots v_t):= \arg \max \{\pref_t(e) \mid e \in \energyTemplate(v_t,c_t)\},
\end{equation}
where $c_t$ is the credit value at time-step $t$. It follows directly from the correctness of optimal 
QaSTels that $\strat_0$ is winning for $(G,\varphi)$. %

A slightly more involved scenario occurs if edges with low preference values are assumed to be blocked 
and hence should not be taken at all by the controlled system. This can be due to a blocking static obstacle perceived by 
a mobile robot, or due to an actuation failure, e.g., a faulty motor in a quad rotor, resulting in a 
restricted evolution of the system and hence in the unavailability of certain $\pz$ moves in the 
game abstraction (which are henceforth assumed to be annotated with preference $0$). In this case, 
\eqref{equ:strat_adapt} changes to
\begin{equation}\label{equ:strat_adapt_faulty}
 \strat_0(v_0\hdots v_t):= \arg \max \{\pref_t(e)>\epsilon\mid e\in \energyTemplate(v_t,c_t)\},
\end{equation}
for some given $\epsilon>0$. Unfortunately, if $\strat_0(v_0\hdots v_t)$  
becomes empty, we cannot continue to control the system with the current template.
However, due to the permissiveness of optimal QaSTels (see~\cref{thm:energyGamePermissive,thm:meanpayoffGamePermissive}), 
such scenarios cannot occur if at least one edge with the minimal activation energy is 
always retained. Formally, we have the following observation.
\begin{proposition}\label{prop:correctnessOfMinAct}
 Given a weighted game graph $(G = (V,E),\weight)$ with an optimal QaSTel $\energyTemplate$, let 
 $E_t^*:=\{e\in E_0~|~\pref_t(e)<\epsilon\}$ for some $t\in \mathbb{N}_0$. If %
 there exists no $v\in V$ s.t.\ $\minActivationEdges(v)\subseteq E_t^* 
 $,%
 where
 \begin{equation}
 \minActivationEdges(v):= \arg\min\{\activation_\energyTemplate(e) 
	\mid e\in E(v) \},
\end{equation}
then $\energyTemplate = \energyTemp(G',\weight,\mu_0)$
 is the optimal QaSTel for $G'\setminus E_t^*$.
\end{proposition}

It follows from \Cref{prop:correctnessOfMinAct} that whenever there exists no $v\in V$ s.t.\
$\minActivationEdges(v)\subseteq E_t^* 
$ for all $t\in \mathbb{N}_0$, the $\pz$ strategy in 
\eqref{equ:strat_adapt_faulty} 
is winning in the original weighted game $(G,\varphi)$. 
Furthermore, if this condition is violated at some time point $t$, a recomputation of QaSTels 
can be triggered. As an obvious corollary (see \cref{sec:correctnessOfHotstarting}) of the 
known monotonicity properties of the value iteration algorithm, this recomputation can be 
hot-started from the current optimal value over $G$, i.e., %
\begin{equation}\label{eq:hotstarting}
  \energyTemplate':=\energyTemp(G',\weight,\mu_0) = \energyTemp(G',\weight,\activation_{\energyTemplate}).
\end{equation}

It should be noted that this dynamic recomputation of QaSTels might return an empty winning region at 
some time step, in which case the dynamically adapted strategy from \eqref{equ:strat_adapt_faulty} 
returns a finite play which is not winning anymore. %
Intuitively, such scenarios occur when preferences and QaSTels do not interact favorably. %
If preferences are due to unmodeled disturbances, such as dynamic obstacles, there is not much one can do to prevent such blocking situations. If additional objectives are, however, known and can be modelled as additional quantitative or qualitative objectives over the given game graph $G$, one should incorporate them into template synthesis as soon as they are available. This then leads to compositional synthesis approaches as discussed next.

%% file: sections/compositionality.tex
\subsection{Composing QaSTels}\label{sec:nocomp}
The previous section has outlined the advantages of QaSTels for the local 
adaptation of strategies at runtime, which is in close analogy to the properties of PeSTels \cite{Anandetal_SynthesizingPermissiveWinning_2023}. Unfortunately, 
this section shows that QaSTels -- in contrast to PeSTels -- are not composable in a straightforward manner. That is,  given a QaSTel $\energyTemplate$ for a weighted game graph $(G,\varphi)$ and a QaSTel $\energyTemplate'$ for a different weighted game $(G,\varphi')$ over the same graph, we cannot easily combine $\energyTemplate$ and 
$\energyTemplate'$ into a QaSTel which is winning for the combined game $(G,\varphi\wedge\varphi')$.
This is due to the fact that a 
winning strategy for $(G,\varphi\wedge\varphi')$ might require infinite memory \cite{multimeanpayoff}.

Nevertheless, we can still extract a single (infinite-memory) winning strategy from multiple QaSTels over mean-payoff games as long as the winning regions of all games coincide. The resulting algorithm, given in \cref{alg:quadtemplate}, uses the function $\combine$ to combine winning strategies of all games extracted from QaSTels following the procedure given in  \cite[Lemma 8]{multimeanpayoff}.

\begin{algorithm}[t]
	\caption{$\textsc{CombineQaSTel}$}\label{alg:quadtemplate}
	\begin{algorithmic}[1]
		\Require Game graph $ \gamegraph=\tup{V, E}$ with $\{\meanpayoff(\weight_i)\}_{i\in[1;k]}$
		\Ensure Winning strategy $\strat$ for $(G,\bigwedge_{i\in[1;k]}\meanpayoff(\weight_i))$
		\State $\win=V$; $ \activation_{\energyTemplate_i} = \mu_0$ for all $i\in[1;k]$
        \State $(\win_i, \energyTemplate_i) \gets \energyTemp(\gamegraph, \weight_i, \activation_{\energyTemplate_i})$
        \State $\win' \gets \bigcap \win_i$
        \While{$ \win\neq\win' $}
        \ForAll{$ e\in \win'\times (\win\setminus\win') $} \do $\energyTemplate(e) = \infty$ \EndFor
        \State $(\win_i, \energyTemplate_i) \gets \energyTemp(\gamegraph, \weight_i, \activation_{\energyTemplate_i})$
        \State $\win\gets\win'$; $\win' \gets \bigcap \win_i$
		\EndWhile
		\State $\strat_i=\textsc{ExtractStrat}(\energyTemplate_i)$, for all $i\in[1;k]$
		\State\Return $ \combine(\{\strat_i\}_{i\in[1;k]})$
	\end{algorithmic}
\end{algorithm}

 \begin{restatable}{theorem}{restateMultiMeanPayoff}\label{thm:multiMeanPayoff}
 Given a game graph $G = (V,E)$ with multiple \emph{mean-payoff} objectives 
	$\{\meanpayoff(\weight_i)\}_{i\in[1;k]}$, $\textsc{CombineQaSTel}(G,\{\meanpayoff(\weight_i)\}_{i\in[1;k]})$ returns a winning strategy for the game $(G,\bigwedge_{i\in[1;k]}\meanpayoff(\weight_i))$. Furthermore, the procedure terminates in time $\bigO(k \cdot \abs{V}\cdot\abs{E}\cdot W)$, where $W$ is the maximal weight in the game. 
\end{restatable}

It should be noted that $\textsc{CombineQaSTel}$ can also be used for iterative 
synthesis, i.e., if a mean-payoff objective arrives, a new combined strategy can be 
derived by hot-starting $\textsc{CombineQaSTel}$.

\begin{remark}
 We remark that games with multiple energy objectives might not always have a winning strategy, even if the winning regions of the different energy objectives coincide. %
\end{remark}

\begin{remark}\label{rem:MeanPayoffParity}
 It is known that parity objectives can be translated into mean-payoff objectives with threshold $0$ over the same game graph in polynomial time \cite[Theorem 40]{gamesongraphsbook}. It therefore follows that the combination of QaSTels and PeSTels might require infinite memory strategies. This further implies that  $\textsc{CombineQaSTel}$ can also be used to extract combined strategies in (multi-objective) mean-payoff parity games.%
\end{remark}

%% file: sections/PeSTeLs.tex
\section{Combining QaSTels with (bounded) PeSTels}\label{sec:meanpayoffCobuchi}

While the previous section discussed control-inspired applications of QaSTels for runtime-adaptability and 
composition of templates for purely \emph{quantitative} objectives, this section considers the construction of 
strategy templates for games with both \emph{quantitative and qualitative} objectives. As \cref{rem:MeanPayoffParity} already 
shows that the most general combination, i.e., mean-payoff-Parity games, require infinite 
strategies in general, we cannot hope for the construction of winning templates which are maximal for the full 
class of parity, and hence, $\omega$-regular objectives over the weighted game graphs $G_\weight$. Instead, we 
propose to start with qualitative strategy templates which under-approximate the set of winning strategies for 
any parity objective, but allow for a straight forward combination with QaSTels, and hence for efficient 
incremental synthesis. These restricted qualitative templates are based on PeSTels from \cite{Anandetal_SynthesizingPermissiveWinning_2023}, which we 
recall in \cref{sec:pestel} before formalizing their composition with QaSTels in \cref{sec:mistels}. \cref{sec:MiSTelApp}
then discusses control-inspired mixed specifications where this class of templates ensure to capture a huge 
class of relevant strategies. We test the completeness and efficiency of the resulting algorithms in 
\Cref{sec:experiments}.

\subsection{Bounded PeSTels}\label{sec:pestel}
Within this paper we consider PeSTels which are composed of two edge conditions:  (i) \emph{unsafe} edges $\safegroup\subseteq E_0$, and (ii)  \emph{co-live} edges $\colivegroup\subseteq E_0$. %
Their combination $\template=(\safegroup,\colivegroup)$ is called a \emph{bounded} PeSTel, which represents the objective $\plays_\template(G) = \{\play \in V^\omega \mid \forall e\in\safegroup:\ e \not\in\play \text{ and } \forall e\in\colivegroup:\ e \not\in \inf(\play)\}$.
We say a strategy $\strat$ (for $\pz$) \emph{follows}  $\template$, denoted by $(G,\strat) \models \template$ (or simply $\strat\models\template$ when $G$ is clear from the context), if $\plays_{\strat}(G) \subseteq \plays_\template(G)$.
Intuitively, $\strat$ follows $\template$ if every $\strat$-play %
(i) never uses the unsafe edges in~$\safegroup$, and
(ii) stops using the co-live edges in~$\colivegroup$ eventually. %
In a qualitative game $(\gamegraph,\spec)$, a PeSTel $\template$ is \emph{winning} from a node $v$ if every strategy following $\template$ is also winning from $v$. 

In \cite{Anandetal_SynthesizingPermissiveWinning_2023}, PeSTels include a third edge condition, called live groups, which ensures that certain edges are taken infinitely often. We discuss in \cref{app:boundedPestels}, how winning PeSTels for games $(\gamegraph,\spec)$ synthesized via the algorithms presented in \cite{Anandetal_SynthesizingPermissiveWinning_2023} can be directly bounded leading to bounded PeSTels. We refer to the combined synthesis algorithm as $\computeTemp$.

\begin{remark}
 We note that $\computeTemp$ usually under-approximates the set of winning plays for $(\gamegraph,\spec)$. It is however known that for co-Büchi games no additional winning strategies can be captured by live-group templates, naturally leading to bounded templates. Notably, the algorithms for computing winning PeSTels in \cobuchi games, presented in~\cite{Anandetal_SynthesizingPermissiveWinning_2023}, exhibit the same worst-case computation time as standard methods for solving such (finite-state) games. %
\end{remark}

\subsection{Mixed Strategy Templates (MiSTels)}\label{sec:mistels}
Following the previous discussion, this section introduces MiSTels $\Strat = 
(\safegroup,\colivegroup,\energyTemplate)$ as a combination of a (bounded) PeSTel $\template = 
(\safegroup,\colivegroup)$ with a QaSTel $\energyTemplate$ defined over the same weighted game 
graph $G_\weight$. Thereby, MiSTels concisely represent a set of winning strategies for \emph{mixed games} $(G,\spec\wedge\varphi)$
which contain a quantitative objective $\spec$ and a qualitative objective $\varphi$. 
A $\pz$ strategy $\strat$ is said to follow the MiSTel $\Strat$ over $G_\weight$ if it follows both $\template$ and $\energyTemplate$ over $G_\weight$. %

\smallskip
\noindent\textbf{Conflict-free MiSTels.}
Given any combination of PeSTels and QaSTels, their direct combination might result in a MiSTel for 
which no strategy exists that follows it. As an example consider the discussion from \cref{sec:LocalAdapt} on the 
runtime adaptation of a strategy which follows a QaSTel but at the same time avoids taking 
unavailable edges. Given a PeSTel $(\safegroup,\colivegroup)$, the edge set $\safegroup$ can 
directly be interpreted as the set of unavailable edges, which -- in contrast to the case discussed 
in \cref{sec:LocalAdapt} -- does not change and is known a priori. Similarly, the set of $\colivegroup$ collects 
all edges that eventually become unavailable. Following \Cref{prop:correctnessOfMinAct}, we 
obtain the existence of a strategy following a MiSTel if $\minActivationEdges(v) 
\not\subseteq \safegroup\cup\colivegroup$ for all $v\in V$. If a MiSTel has this property, we call it \emph{conflict 
free}. %
Similar to \Cref{prop:extractStrategy}, one can extract a 
strategy following a conflict-free MiSTel by picking an unconstrained edge (i.e.\ $e\notin\safegroup\cup\colivegroup$) with the smallest activation value at every node.
\begin{restatable}{proposition}{restatePextractMixedStrategy}\label{prop:extractMixedStrategy}
	Given a weighted game graph $\gamegraph_\weight$ with a conflict-free MiSTel $ \Strat = (\safegroup,\colivegroup,\energyTemplate)$, a positional strategy following $ \Strat $ can be extracted in time $ \bigO(\abs{E}) $.
\end{restatable}

\smallskip
\noindent\textbf{Winning MiSTels.}
We say that the MiSTel $\Strat$ is \emph{winning} in the \emph{mixed game} 
$(G,\spec\wedge\varphi)$ from a node $v$ if all strategies $\strat$ that follow 
$\Strat$ are winning from $v$ in both the quantitative game $(G,\varphi)$ and the qualitative game 
$(G,\spec)$. %
\begin{algorithm}[t]
	\caption{$ \mixedTemp(\gamegraph, \weight,\spec) $}\label{alg:mixedtemplate}
	\begin{algorithmic}[1]
		\Require Mixed game $ (\gamegraph=\tup{V, E},\varphi\wedge\spec)$ with $\varphi = \meanpayoff(\weight)$ or $\energy(\weight)$
		\Ensure winning region $\win$, winning conflict-free MiSTel $\Strat$  
		\State $\activation_\energyTemplate = \mu_0$
        \State $(\win, \conflict,\Strat) \gets \FindConflicts(\gamegraph, \weight,\spec, \energyTemplate) $
        \While{$ \conflict \neq \emptyset $}
			\State $\spec \gets \spec \wedge \safetygame(\win)$\label{line:alg:mixedTemp:pestelAddSafety}
			\ForAll{$ e\in \conflict $} \do $\energyTemplate(e) = \infty$ \EndFor\label{line:alg:mixedTemp:qastelEdgeInf}
			\State $ (\win, \conflict,\Strat) \gets \FindConflicts(\gamegraph, \weight, \spec, \energyTemplate) $
		\EndWhile
		\State\Return $ (\win, \Strat) $
		
		\Statex
		
		\Procedure{\FindConflicts}{$ \gamegraph, \weight,\spec, 
		\energyTemplate $}
            \State $(\win_\spec,(\safegroup,\colivegroup)) \gets \computeTemp(\gamegraph, \spec)$\label{line:alg:mixedTemp:pestelCompute}
            \State $(\win_\varphi, \energyTemplate) \gets \energyTemp(\gamegraph, \weight, \activation_\energyTemplate)$\label{line:alg:mixedTemp:qastelCompute}
            \State $\win \gets \win_\spec\cap \win_\varphi$
            \State $\conflict = \cup_{v\in\win} \{\minActivationEdges(v) \mid \minActivationEdges(v)\subseteq \safegroup\cup\colivegroup\}$\label{line:alg:mixedTemp:conflictsCompute}
			\State \Return $ (\win,\conflict,(\safegroup,\colivegroup,\energyTemplate)) $
		\EndProcedure
	\end{algorithmic}
\end{algorithm}
In order to synthesize a winning MiSTel for a given mixed game $ (\gamegraph=\tup{V, E},\varphi\wedge\spec)$, we can therefore iteratively construct winning PeSTels and QaSTels and remove all conflicts from their joint winning region. An efficient way to do so is formalized in \cref{alg:mixedtemplate}. After QaSTel synthesis is initialized with~$\mu_0$, winning PeSTels and QaSTels are computed (\cref{line:alg:mixedTemp:pestelCompute,line:alg:mixedTemp:qastelCompute}) and conflicts in their joint winning region are detected (\cref{line:alg:mixedTemp:conflictsCompute}). If no conflict is detected, the algorithm directly terminates. Otherwise, conflicts are resolved by (i) adding an additional safety requirement to the quantitative specification $\spec$ (\cref{line:alg:mixedTemp:pestelAddSafety}), and (ii) increasing the edge-weight of conflicting edges, i.e., edges in $\minActivationEdges(v)\subseteq \safegroup\cup\colivegroup$, to $\infty$ (\cref{line:alg:mixedTemp:qastelEdgeInf}). After this, PeSTels and QaSTels are recomputed to resolve conflicts, which might result in new ones. 
If no more conflicts are generated, the algorithm terminates. 
The resulting MiSTel is winning and conflict free, as formalized next.

\begin{restatable}{theorem}{restateMixedTemplateThm}\label{thm:mixedTemplate}
	Let $\game = (\gamegraph,\varphi\wedge\spec)$ be a mixed game with $\varphi = 
	\meanpayoff(\weight)$ or $\energy(\weight)$ and $\spec$ a qualitative objective. Then, if $ 
	(\win, \Strat) =  \mixedTemp(\gamegraph,\weight,\spec)$, it holds that $ \Strat$ is a 
	conflict-free winning MiSTel from $\win$.
\end{restatable}

\begin{remark}
   We note that the iterative computation of PeSTels and QaSTels in \cref{alg:mixedtemplate} 
   can be hot-started, which makes $\mixedTemp$ more efficient. For QaSTels, the correctness of 
   hot-starting follows from  \cref{eq:hotstarting}. Notably, $\computeTemp$ can also be 
   hot-started if specifications are added 
   \cite[Alg.4]{Anandetal_SynthesizingPermissiveWinning_2023}.
 \end{remark}

\smallskip
\noindent\textbf{Incremental MiSTel Synthesis.}
Surprisingly, \cref{alg:mixedtemplate} can directly be extended to incremental MiSTel synthesis. That is, given an already computed winning MiSTel $\Strat=(\safegroup,\colivegroup,\energyTemplate)$ with winning region $\win$ for the mixed game $\game = (\gamegraph,\varphi\wedge\spec)$, $\Strat$ can be refined to a new winning MiSTel $\Strat'=(\safegroup',\colivegroup',\energyTemplate')$ for the mixed game $\game = (\gamegraph,\varphi\wedge\spec\wedge\spec')$ with winning region $\win'\subseteq\win$ if a new quantitative objective $\spec'$ arrives. 
For this, one would use \cref{alg:mixedtemplate} for the combined quantitative objective $\spec\wedge\spec'$ and hot-starts both $\computeTemp$ and $\energyTemp$ with $(\safegroup,\colivegroup)$ and $\energyTemplate$ from the already existing MiSTel $\Strat$. %

\begin{remark}
 We recall from \cref{sec:nocomp} that adding additional \emph{quantitative} objectives only allows to extract new winning \emph{strategies} (and no templates) if all quantitative objectives are mean-payoff objectives. Nevertheless, this clearly can also be incorporated in an iterative version of \cref{alg:mixedtemplate}. %
\end{remark}

\subsection{Applications of MiSTels}\label{sec:MiSTelApp}
While MiSTels can be applied to any mixed game, the class of winning strategies they capture might be a quite restricted (possibly empty) subset of all available strategies. This can (i) downgrade the adaptability of strategy choices during runtime %
and (ii) might lead to an empty winning region in incremental or compositional synthesis approaches. MiSTels therefore have a higher potential whenever they capture a large set of winning strategies. One such example are co-Büchi games, where it is known that PeSTels are naturally bounded. After discussing how specifications commonly used in CPS applications reduce to co-Büchi games in this section, we investigate this instance further in \cref{sec:MPCB_Compare}. %

\smallskip
\noindent\textbf{LTL$_f$ Specifications.}
LTL$_f$ is a fragment of linear temporal logic (LTL) where system properties are evaluated over finite traces only, which has gained popularity in robotic applications over the last decade \cite{de2015synthesis}. Specifications given in LTL$_f$ can be translated into deterministic finite automata which can be composed with a weighted game, adding a reachability objective to it.  
The resulting reachability objective can then be translated into a co-Büchi objective by removing all outgoing transitions from target states, adding self-loops  to them and adding all states which are not a target state into the co-Büchi region.    

\smallskip
\noindent\textbf{Uniform Attractivity.}
If one manipulates the graph as outlined before, one can show that for every winning strategy $\strat$ exists a time bound $k\in\N$ s.t.\ all plays $\play$ compliant with $\strat$ will never leave the set of target states again after $k$ time steps, i.e., $\forall i\geq k.\ \play[i]\in T$ \footnote{For classical co-Büchi games, such a uniform bound over all plays does in general not exist (see \cite{GIRARD2021109543} for an example).}. %
Such specifications are called uniform attractivity specifications and are used to translate classical stability objectives into formal specifications. Formally, the computation of winning strategies for general uniform attractivity games (without manipulating the game graph) are solved by first computing the winning set of the safety objective $\win_{safe}:=Safety(T)$ and then solving the co-Büchi game $(G,\cobuchi(\win_{safe}))$. It is therefore not hard to see, that PeSTels for uniform attractivity games are also naturally bounded.

\smallskip
\noindent\textbf{Adding Quantitative Objectives.}
Due to the fact that every winning strategy comes with a uniform bound on when the target set is 
reached, the above instances of co-Büchi games are naturally combined with qualitative energy 
objectives with fixed initial credit. This is in contrast to `classical' co-Büchi objectives which 
are more naturally combined with mean-payoff objectives, ensuring that strategies are optimal 
in the limit even if they deviate for finite time durations. %

\subsection{Mean-Payoff co-Büchi Games}\label{sec:MPCB_Compare}
Mean-payoff co-Büchi games have been already studied in \cite{ChatterjeeHS17}
and will therefore be used to benchmark our prototype implementation of MiSTel synthesis against the state of the art algorithm, called \othertool, in \cref{sec:experiments}. %
 We note that these results carry over to control-inspired applications of MiSTels discussed before, as QaSTel synthesis coincides for energy and mean-payoff objectives and reachability and uniform attractivity simply utilize the (bounded) PeSTel synthesis algorithm specialized for co-Büchi games called $\cobuchiTemp$ from \cite[Alg.2]{Anandetal_SynthesizingPermissiveWinning_2023}. 

\smallskip
\noindent\textbf{Complexity.}
Using $\cobuchiTemp$ from \cite[Alg.2]{Anandetal_SynthesizingPermissiveWinning_2023} instead of $\computeTemp$ in $\mixedTemp$ (\cref{alg:mixedtemplate}), the complexity of MiSTel synthesis for mean-payoff co-Büchi games reduces to $ \bigO(n^2m+nmW) $, where $ n = \abs{V} $, $ m = \abs{E} $, and $ W$ is the maximum weight in $\weight$ (see \cref{corr:MPcobuchi} in \cref{sec:MPcobuchi}).
In comparison, the worst-case computation time of \othertool is $ \bigO(nmW) $ \cite[Theorem 5]{ChatterjeeHS17} and, hence, lower. %
We however note that we achieve the same worst-case complexity when $ W \geq n $, which is very often the case. %

\smallskip
\noindent\textbf{Completeness.}
We note that $\mixedTemp$ is not complete. First of all, the bounded PeSTel computed by $\cobuchiTemp$ does not capture winning strategies for instances %
in which $\po$ (unexpectedly) helps $\pz$ and is therefore already incomplete (see \cite{Anandetal_SynthesizingPermissiveWinning_2023} for details). In addition, conflicts with co-Büchi nodes are removed immediately for conflict resolution in $\mixedTemp$, while they could be used by a strategy finitely often. This leads to a further potential under-approximation of the winning region, as illustrated in the example game depicted in \Cref{fig:incomplete}. Here, \cobuchiTemp outputs the co-live edges denoted by orange dashed lines. Further, the activation energy of the edges $ (a,b) $ and $ (a,a) $ are $0$ and $1$ respectively, and hence, the edge $ (a,b) $ is a conflict by definition. Hence, in \cref{alg:mixedtemplate}, the edge $ (a,b) $ will be assigned the value $ \infty $ while resolving the conflict, making the node $ a $ losing in the next iteration. However, we observe that all the nodes are winning for $ \pz $.

\begin{wrapfigure}[8]{r}{0.45\textwidth} 
	\vspace{-1cm}
	\centering
	\begin{tikzpicture}[shorten >=1pt, node distance=2cm, on grid, auto, 
		unsafe/.style={red!70, dotted}, %
		colive/.style={orange!70, dashed}]
		\node[state, accepting] (a) {$ a $};
		\node[state] (b) [right of=a] {$ b $};
		\node[state,accepting] (c) [right of=b] {$ c $};
		
		\path[->] (a) edge [loop below] node {-1} (a) 
					edge [bend left, colive] node {0} (b);
		\path[->] (b) edge [bend left] node {-1} (a)
					edge [bend left] node {0} (c);
		\path[->] (c) edge [bend left,colive] node {-1} (b)
					edge [loop below] node {0} ();
	\end{tikzpicture}
	\vspace{-0.2cm}
	\caption{Mean-payoff \cobuchi game with $ \cobuchi(\{a, c\}) $
	.}\label{fig:incomplete}
\end{wrapfigure}
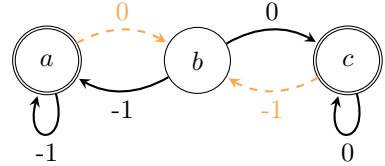

We note that the state-of-the-art algorithm by Chatterjee et al. \cite{ChatterjeeHS17} is instead complete. Our experimental results presented in \cref{sec:experiments} however show that the winning region computed by \mixedTemp coincides with the full winning region computed by \othertool for more then $90\%$ of the considered benchmark instances.

%% file: sections/experiments.tex
\section{Empirical Evaluations}~\label{sec:experiments}
This section aims to highlight the advantages of strategy templates for games with qualitative (and 
quantitative) objectives. However, benchmarking QaSTel (and MiSTel) synthesis as well as their adaptability 
and compositional properties is difficult for two reasons. First, to the best of our knowledge, 
there are no benchmark suites that include real-world applications with combined quantitative and 
qualitative objectives. Second, while there exist algorithms for mean-payoff parity games, we are 
not aware of any implementation or benchmark-based evaluation of these algorithms. To address this, 
we build new benchmark suites based on the SYNTCOMP benchmark~\cite{syntcomp} and implement the 
\othertool algorithm from \cite{ChatterjeeHS17} (discussed in \cref{sec:MPCB_Compare}) to benchmark 
our Java-based prototype implementation \tool of \mixedTemp against it.

\smallskip
\noindent\textbf{Experimental Setup.} 
We built a new benchmark suite from the SYNTCOMP benchmark~\cite{syntcomp} by (a) translating 
SYNTCOMP parity games into mean-payoff games using standard techniques 
\cite{MeyerLuttenberger_gpumpg}, and (b) adding co-Büchi objectives to the qualitative games 
from 
(a) by randomly choosing avoidance regions (i.e., the set of co-Büchi nodes which a play is not 
allowed to visit infinitely often). For practical reasons, 
we imposed limits of $5 \times 10^5$ nodes and $10^5$ energy values on the edges.
As a 
result, our benchmark suite comprises 245 mean-payoff game graphs.
All experiments were executed on a 32-core Debian machine equipped with an Intel Xeon E5-V2 CPU 
(3.3 GHz) and up to 256 GB of RAM.

\begin{figure}[t]\label{plots}
	\centering
	\begin{subfigure}[b]{0.45\textwidth}
		\includegraphics[width=\textwidth]{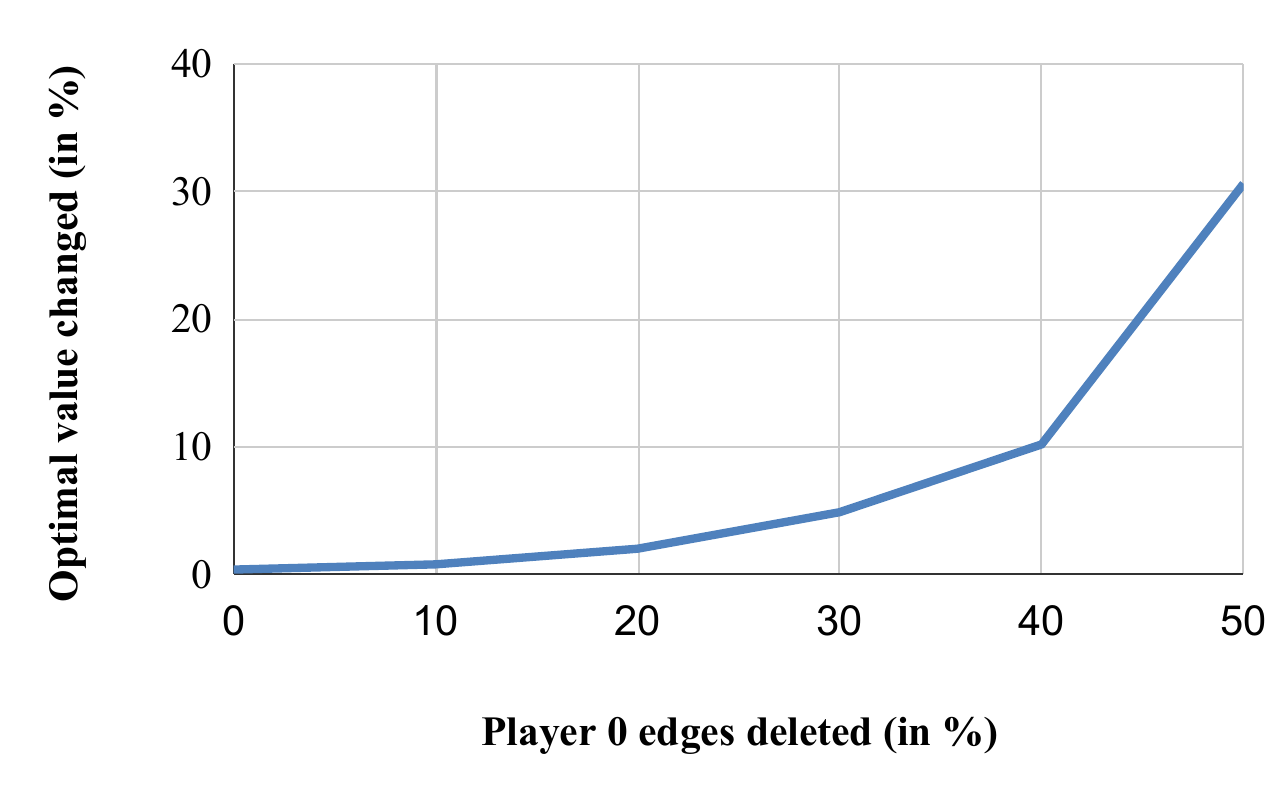}
		\caption{Fault tolerance}\label{plots_faultTolerance}
	\end{subfigure}
	\quad
	\begin{subfigure}[b]{0.45\textwidth}
		\includegraphics[width=\textwidth]{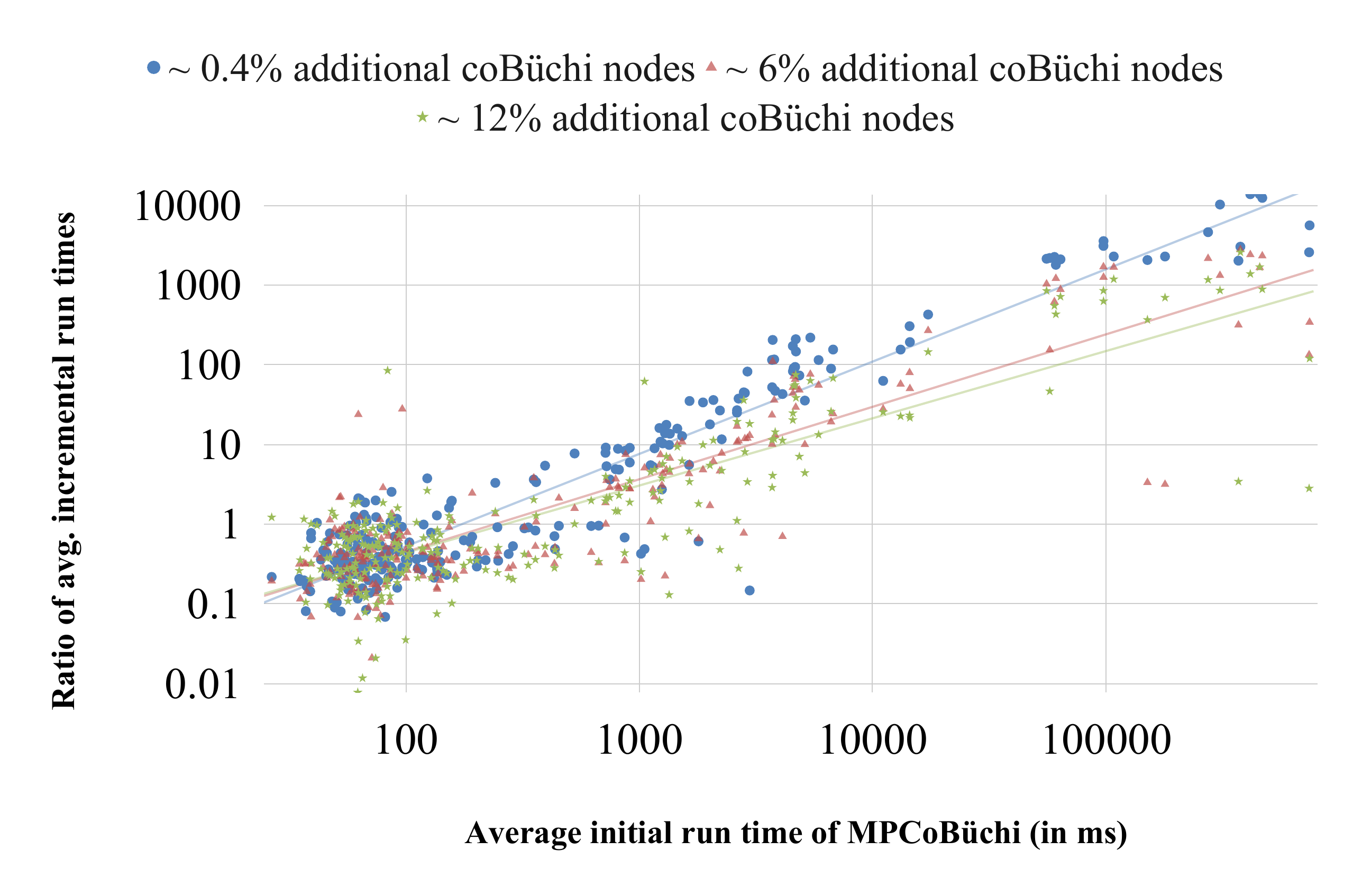}
		\caption{Incremental synthesis}\label{plots_incremental}
	\end{subfigure}

	\begin{subfigure}[b]{0.45\textwidth}
		\includegraphics[width=\textwidth]{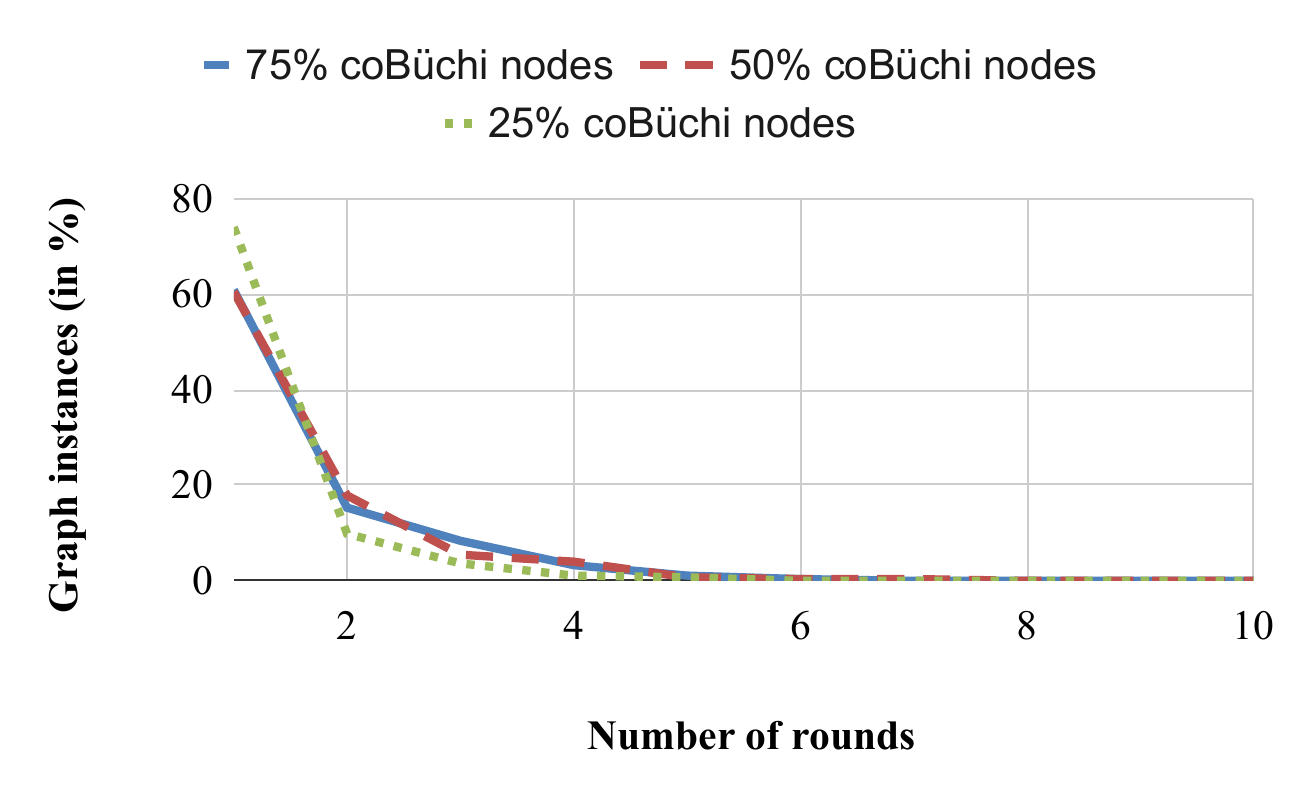}
		\caption{Rounds of conflict resolution}\label{plots_allMPCoBuechi}
	\end{subfigure}
	\quad
	\begin{subfigure}[b]{0.45\textwidth}
		\includegraphics[width=\textwidth]{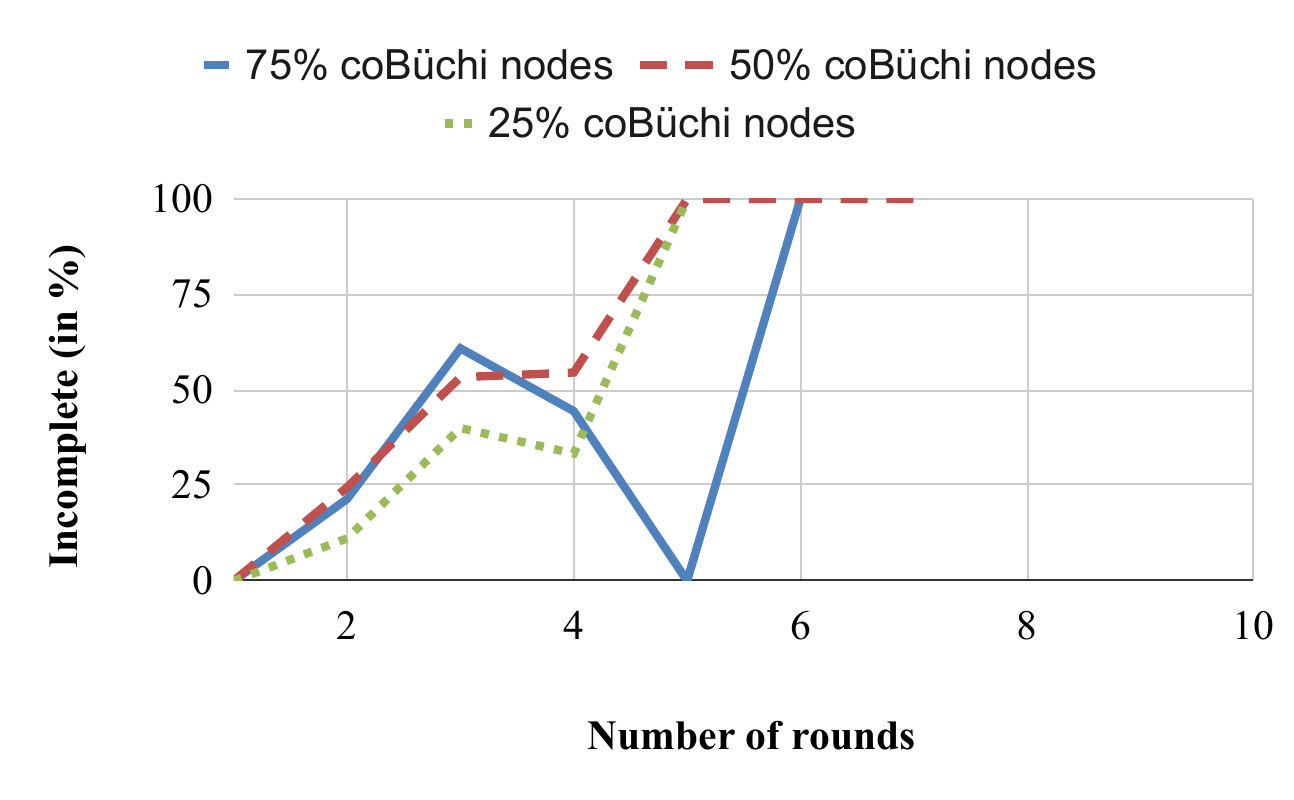}
		\caption{Incompleteness of \tool}\label{plots_incomplete}
	\end{subfigure}
	
	\caption{Plots summarizing the experimental evaluations. Bigger figures can be found in the 
	appendix. }
	\vspace{-2em}
\end{figure}

\smallskip
\noindent\textbf{Dynamic Strategy Adaptation.}
We have conducted experiments to evaluate the robustness of QaSTels to the unavailability of edges due to additional edge preferences used for dynamic strategy extraction at runtime. As discussed in \cref{sec:LocalAdapt}, QaSTels do not need to be adapted if the edges with the minimum activation energy are still available. %
To evaluate how likely QaSTels need to be recomputed we randomly removed edges from the game graphs in our benchmark suite and measured the average number of deletions required until a minimum activation edge was removed for any node. We 
ran \tool{} on each graph, incrementally deleting randomly selected edges until a change in 
the optimal value occurred. This process was repeated 10 times per graph, and we computed the 
average number of edge deletions needed to trigger this change. 
\cref{plots_faultTolerance} illustrates the 
trend between 
the proportion of game graphs in which the optimal value changes against the proportion of \pz 
edges removed from them. The trend clearly shows that in practice, it is very unlikely that 
minimum activation edges are deleted even after removing a significant number of edges 
from 
the graph. This establishes the efficiency of QaSTels to produce robust strategies.

\smallskip
\noindent\textbf{Incremental Synthesis.} 
We have benchmarked our prototype implementation \tool of \mixedTemp for mean-payoff \cobuchi games 
against our implementation of the \othertool algorithm \cite{ChatterjeeHS17}.
The experiment starts by providing both algorithms only with a mean-payoff game from our benchmark suite.
We then add 5 \cobuchi objectives incrementally. To evaluate the dependence of runtimes on the \cobuchi objectives, we run this experiment thrice, 
with varying amounts of additional nodes added to the avoidance region in every incremental step: (i) 0.4\% (blue circles in \cref{plots_incremental}), (ii) 6\% (red triangles in \cref{plots_incremental}) and (iii) 12\% (green stars in \cref{plots_incremental}) leading to (i) ~2\%, (ii) ~30\% and (iii) ~60\% avoidance region in the final iteration of incremental synthesis. %
To compare the performance of \tool and \othertool, \cref{plots_incremental} shows (i) the ratio of average runtimes of \othertool vs \tool to complete the 5 incremental synthesis steps outlined above (y-axis in \cref{plots_incremental}), against (ii) the running times of \othertool over the original mean-payoff game without co-Büchi objectives (x-axis in \cref{plots_incremental}).

The plot shows that when the instances are very simple (resulting in a low 
initial runtime 
of \othertool), \tool may be slow in the recomputation on additional \cobuchi objectives. However, as 
the graphs get more complex, \tool is magnitudes faster than \othertool in recomputing the winning 
regions on the fly. In fact, we see that \tool is around 10000 times faster on the most complex 
instances. This demonstrates the utility of our approach avoiding the need of recomputations from 
scratch. 

\smallskip
\noindent\textbf{Completeness.} 
We conducted experiments to quantify the loss of completeness (in terms of the size of the winning region) of $\mixedTemp$ compared to \othertool.  
For each of the 245 mean-payoff game graphs in our benchmark suite, we randomly 
selected (i) 25\%, (ii) 50\%, and (iii)75\% of the nodes and defined them as the avoidance region 
in the respective mean-payoff co-Büchi game. \tool could \emph{not} compute the \emph{complete} 
winning region in (i) 10, (ii) 30 and (iii) 28 instances of games, respectively. This shows that 
$\tool$ computes the full winning region for over $90\%$ of the considered games. 

To further investigate the root cause of incompleteness, we evaluated the number of rounds of 
conflict resolutions required by \tool for each instance. In \cref{plots_allMPCoBuechi}, we plot the 
percentage of instances requiring a certain number of conflict resolution rounds. In 
\cref{plots_incomplete}, we plot the percentage of incomplete instances against the number of 
conflict resolution rounds required to solve the respective instance. 
Together the two plots present the relationship between the number of game graphs for 
which \tool{} was unable to synthesize a winning strategy and the number of iterations 
required to resolve conflicts. \cref{plots_incomplete} indicates a clear trend: as the number of 
iterations 
increases, the likelihood of failure to synthesize a winning strategy also increases. However, 
as noted in \cref{plots_allMPCoBuechi}, the number of iterations required in practice remains low. 
This observation explains the low number of incomplete instances (less than 10\%, as noted above), 
and supports the claim that while 
our algorithm may be incomplete in the worst case, it is able to synthesize winning strategies in 
most scenarios.

%% file: sections/appendix.tex
\appendix
\section{Missing Proofs}\label{sec:missingproofs}

\subsection{Proof of \Cref{prop:extractStrategy}}
\restatePextractStrategy*
\begin{proof}
	Consider the positional strategy $\strat$ such that for every $v\in V_0$, $\strat(v)$ is an edge with the smallest activation value in $\energyTemplate$, i.e., $\strat(v) = \arg\min\{\activation_\energyTemplate(e) \mid e\in E(v)\}$. 
	We will show that $\strat$ follows $\energyTemplate$ in every weighted game with weight function $\weight$.
	It is enough to show that $(G,\strat) \vDash_{c}\energyTemplate$ for every $c\in\N$.
	
	Let $c\in\N$ and consider a $\strat$-play $\play = v_0v_1\cdots$.
	For each $i\in\N$, let $e_i = (v_i,v_{i+1})$ and $c_i = c + \weight(\play[0;i])$ be the edge and the credit after $i$-th step.
	Let $k\in\Ninf$ be the smallest index such that $v_k\in V_0$ and $e_k\notin \energyTemplate(v_k, c_k)$.
	Then, it is enough to show that $\energyTemplate(v_k, c_k) = \emptyset$.
	Suppose $e\in \energyTemplate(v_k, c_k)$, then $\activation_\energyTemplate(e) \leq c_k$.
	As $\strat(v_k)=e_k$ is the edge with the smallest activation value at $v_k$, it holds that $\activation_\energyTemplate(e_k) \leq \activation_\energyTemplate(e) \leq c_k$ and hence, $e_k\in \energyTemplate(v_k, c_k)$. 
	This contradicts the choice of $k$ and hence, $\energyTemplate(v_k, c_k) = \emptyset$.
	\qed
\end{proof}

\subsection{Proof of \Cref{thm:Winning}}
\restateWinning*
We will show the proof for energy games and mean-payoff games separately in the following two lemmas.

\begin{lemma}\label{thm:energyGameWinning}
	Given a game graph $G = (V,E)$ with weight function $\weight$, the optimal QaSTel $ \energyTemplate $ is winning in the energy game $(G,\energy_c(\weight))$ for every initial credit $c\in \N$.
\end{lemma}

\begin{proof}
	Consider an energy game $(G,\energy_c(\weight))$ for some initial credit $c\in \N$, and let $u_0\in V$ be a node in the winning region.
	We need to show that for every strategy $\strat$ following $\energyTemplate$ is winning from $u_0$, i.e., every $\strat$-play from $u_0$ is winning.
	As every such $\strat$-play from $u_0$ is a $(\energyTemplate,c)$-play from $u_0$, it is enough to show that every $(\energyTemplate,c)$-play from $u_0$ is winning.
	
	Let $\play = u_0u_1\cdots$ be a $(\energyTemplate,c)$-play from $u_0$. 
	For each $i\in\N$, let us denote by $e_i = (u_i,u_{i+1})$ the edge taken at step $i$, and the credit at the beginning of step $i$ by $c_i = c + \weight(\play[0;i])$.
	By definition of $\energyTemplate$, there exists a $k\in \Ninf$ such that for all $i\in [0;k]$ with $u_i\in V_0$, $e_i\in \energyTemplate(u_i,c_i)$, and if $k\neq \infty$, then $u_{k+1}\in V_0$ with $\energyTemplate(u_{k+1},c_{k+1}) = \emptyset$.
	
	Suppose $k \neq \infty$. As $e_k\in \energyTemplate(u_k,c_k)$, by definition of $\energyTemplate$, $c_k\geq \optE(e_k)$. Hence, there exists a $\pz$ strategy $\strat$ such that every $\plays_\strat(G,e_k)$ is winning for initial credit $c_k$. That means, every play that starts with $e_k$ and then follows $\strat$ is winning for initial credit $c_k$. 
	This implies every $\strat$-play from $u_{k+1}$ is winning for initial credit $c_k-\weight(e_k) = c_{k+1}$.
	Consider the strategy $\strat'$ such that for every history $\prefix\in V^*V_0$, $\strat'(\prefix) = \strat(u_{k+1}\prefix)$.
	Then, by construction, every play that starts with the edge $e_{k+1}' = \strat(u_{k+1})$ (this is well-defined as $u_{k+1}\in V_0$) and then follows $\strat'$ is winning for initial credit $c_{k+1}$.
	Thus, by definition of $\optE$, $c_{k+1}\geq \optE(e_{k+1}')= \activation_\energyTemplate(e_{k+1}')$, and hence $e_{k+1}'\in \energyTemplate(u_{k+1},c_{k+1})$. This contradicts the assumption $\energyTemplate(u_{k+1},c_{k+1}) = \emptyset$. Hence, $k = \infty$.
	
	As $k = \infty$, it holds that for every $i\in \N$, $e_i\in \energyTemplate(u_i,c_i)$, and hence, by definition of QaSTels, $c_i\geq 0$.
	Therefore, $\play$ is winning. \qed
\end{proof}

\begin{lemma}\label{thm:meanpayoffGameWinning}
	Given a game graph $G = (V,E)$ with weight function $\weight:E\rightarrow [-W;W]$, the optimal QaSTel $ \energyTemplate $ is winning in the mean-payoff game $(G,\meanpayoff(\weight))$.
\end{lemma}
\begin{proof}
	Let $u_0\in\win(G,\meanpayoff(\weight))$ be a node in the winning region and
	let $\strat$ be a strategy that follows $\energyTemplate$ in mean-payoff game $(G,\meanpayoff(\weight))$.
	Then, it holds that $\strat\vDash_{c}\energyTemplate$ for some $c \geq W\cdot\abs{V}$.
	We need to show that every $\strat$-play from $u_0$ is winning.
	
	It is well-known that the winning region of mean-payoff games is same as the winning region of energy games with unknown initial credit~\cite{fastermeanpayoffgames}, i.e., $\win(G,\meanpayoff(\weight)) = \win(G,\energy(\weight))$. Furthermore, as $\win(G,\energy(\weight)) = \win(G,\energy_{c}(\weight))$, we have $u_0\in \win(G,\energy_{c}(\weight))$.
	Let $\play$ be a $\strat$-play from $u_0$.
	As $\strat\vDash_{c}\energyTemplate$, by \cref{thm:energyGameWinning}, $\play$ is winning in the energy game $(G,\energy_{c}(\weight))$. That means, for each $i\in\N$, $c + weight(\play[0;i])\geq 0$.
	This implies that for each $i\in\N$, $\average(\play[0;i]) = \frac{\weight(\play[0;i])}{i} \geq \frac{-c}{i}$.
	Thus, $\limsup_{i\to\infty}\average(\play[0;i])\geq \lim_{i\to\infty}\frac{-c}{i} = 0$, and hence, $\play$ is winning in the mean-payoff game $(G,\meanpayoff(\weight))$.\qed
\end{proof}

\subsection{Proof of \Cref{thm:energyGamePermissive}}
\restateEnergyMaximal*
\begin{proof}
	Consider an energy game $(G,\energy_c(\weight))$ for some initial credit $c\in \N$ and let $\strat$ be a winning strategy.
	It is enough to show that for every $\strat$-play is a $(\energyTemplate,c)$-play.

	Let $\play = u_0u_1\cdots$ be a $\strat$-play. For each $i\in \N$, let $e_i = (u_i,u_{i+1})$ be the edge taken at step $i$, and the credit at the beginning of step $i$ by $c_i = c + \weight(\play[0;i])\geq 0$.
	For every $i\in \N$, let $\strat^i$ be a $\pz$ strategy such that for every history $\prefix$, $\strat^i(\prefix) = \strat(\play[0;i]\prefix)$.

	Suppose $\play$ is a winning play,
	then, by construction, every $\strat^i$-play from $u_{i+1}$ is winning for initial credit $c_{i+1}$. This implies that every play that starts with $e_i$ and then follows $\strat^i$ is winning for initial credit $c_{i+1}-\weight(e_i) = c_i$.
	Hence, by definition of $\optE$, $c_i\geq \optE(e_i) = \activation_\energyTemplate(e_i)$, and hence, $e_i\in \energyTemplate(u_i,c_i)$ for all $i\in \N$. Therefore, $\play$ is a $(\energyTemplate,c)$-play. 
	
	Now, suppose $\play$ is not winning. As $\strat$ is a winning strategy and $\play$ is a $\strat$-play, it holds that $u_0\not\in\win(G,\energy_c(\weight))$. This implies, there is no winning strategy from $u_0$ for initial credit $c$. 
	This means, for every edge $e\in u_0E$, there is no strategy $\strat'$ such that $\plays_{\strat'}(G,e) \subseteq \energy_c(\weight)$, and hence, $\optE(e) > c$.
	Therefore, $\energyTemplate(u_0,c) = \emptyset$, and hence, $\play$ is a $(\energyTemplate,c)$-play.
	\qed
\end{proof}

\subsection{Proof of \Cref{lemma:meanpayoffFiniteMemory}}
\restateLMeanpayoffFiniteMemory*
\begin{proof}
	Let $\strat$ be a winning strategy with finite memory in the mean-payoff game $(G,\meanpayoff(\weight))$. Let $v\in\win(G,\meanpayoff(\weight))$ be a winning node and let $\play$ be a $\strat$-play from $v$.
	
	Consider the game graph $G_\strat = (V',E')$ induced by the strategy $\strat$.
	Then, every simple cycle reachable from $v$ in the game graph $G_\strat$ has a non-negative weight, the $\strat$-play from $v$ that just loops in the cycle with negative weight has a negative limit average weight.
	Moreover, as $\play$ is a play in $G_\strat$, after an initial acyclic part of length $< \abs{V'}$, the play $\play$ only visits simple cycles of $G_\strat$.
	As all such cycles have non-negative weight, the total weight of every prefix of $\play$ must be at least $-2\abs{V'}\cdot W$, i.e., $\weight(\play[0;i])\geq - 2\abs{V'}\cdot W$ for all $i\in \N$.
	As $v_0$ and $\play$ are arbitrary, for weight bound $\bound{\strat} = 2\abs{V'}\cdot W$, the lemma holds.\qed
\end{proof}

\subsection{Proof of \Cref{thm:meanpayoffGamePermissive}}
\restateTmeanpayoffGamePermissive*
\begin{proof}
	Let $\strat$ be a winning strategy with finite memory in the mean-payoff game $(G,\meanpayoff(\weight))$.
	We need to show that every $\strat\vDash_{c}\energyTemplate$ for some $c \geq W\cdot\abs{V}$, where $W$ is the maximum weight in the game.
	
	By \cref{lemma:meanpayoffFiniteMemory}, there exists a weight bound $\bound{\strat}\in \N$ such that for every $\strat$-play $\play$ from a winning node $v\in\win(G,\meanpayoff(\weight))$, it holds that $\weight(\play[0;i])\geq -\bound{\strat}$ for all $i\in \N$.
	Let $c = \max\{W\cdot\abs{V},\bound{\strat}\}$.
	Since $\win(G,\meanpayoff(\weight)) = \win(G,\energy_{c}(\weight))$ (as discussed in the proof of \cref{thm:meanpayoffGameWinning}), it holds that for every $\strat$-play $\play$ from a node $v\in\win(G,\energy_{c}(\weight))$, we have $c + \weight(\play[0;i])\geq \bound{\strat} + \weight(\play[0;i])\geq 0$ for all $i\in \N$. Hence, by definition, $\strat$ is a winning strategy in the energy game $(G,\energy_{c}(\weight))$.
	Then, by \cref{thm:energyGamePermissive}, $\strat\vDash_{c}\energyTemplate$.\qed
\end{proof}

\subsection{Proof of \Cref{thm:multiMeanPayoff}}
\restateMultiMeanPayoff*
\begin{proof}
	Let $\win^j$ and $\win^j_i$ be the corresponding region, and $\energyTemplate^j_i$ be the corresponding QaSTel in the $j$-th iteration of the while loop. 
	As $\win^j_i$ is a winning region computed by $\energyTemp$, there is no $\po$ edge from $\win^j_i$ to outside of $\win^j_i$.
	Hence, all the edges from $W^j$ to outside of $\win^j$ are $\pz$ edges.
	Furthermore, since in every iteration, we are hot-starting the procedure $\energyTemp$ from previous iteration by making activation values of some of these $\pz$ edges $\infty$, the winning region $\win^{j+1}$ is a subset of $W^j$.
	Hence, the while loop terminates within $\abs{V}$ iterations. As the algorithm only hot-starts the value iteration algorithm for each objective, the time complexity of the algorithm is $O(k\cdot\abs{V}\cdot\abs{E})$.
	
	Now, suppose $k$ is the last iteration. Then, as $\energyTemplate^k_i$ is obtained from $\energyTemp(G,\weight_i,\cdot)$, it holds that $\energyTemplate^k_i$ is winning from $\win^k$ in the mean-payoff game $(G,\meanpayoff(\weight_i))$.
	Hence, every strategy $\strat_i$ obtained in the algorithm after termination of while loop is a winning strategy from $\win^k$ in the mean-payoff game $(G,\meanpayoff(\weight_i))$.
	Then, by the property of $\combine$, $\combine(\strat_1,\ldots,\strat_k)$ is a winning strategy from $\win^k$ in the mean-payoff game $(G,\bigwedge_{i\in[1;k]}\meanpayoff(\weight_i))$.
	Hence, it is enough to show that $\win^k$ is the winning region $\win'$ of the mean-payoff game $(G,\bigwedge_{i\in[1;k]}\meanpayoff(\weight_i))$.

	By the previous argument, $\win^k\subseteq \win'$.
	Now, let us show using induction on $j\in[1;k]$ that $\win^j\supseteq \win'$.
	For base case $j=1$, as intersection of winning regions of each mean-payoff objective contains the winning region of their conjunction, it holds that $\win^1\supseteq \win'$.
	Now, suppose $\win^j\supseteq \win'$ for some $j\in[1;k-1]$. Since removing edges going out of winning region does not change the winning region, making activation values of some of these edges $\infty$ does not change the winning region. Hence, $\win^{j+1}\supseteq \win'$. Therefore, $\win^k = \win'$.\qed
\end{proof}

\subsection{Proof of \Cref{prop:correctnessOfHotstarting}}\label{sec:correctnessOfHotstarting}
\begin{corollary}\label{prop:correctnessOfHotstarting}
	Given the premises of Prop.~\ref{prop:correctnessOfMinAct}, it holds that the \qastel $\energyTemplate':=\energyTemp(G',\weight,\mu_0) = \energyTemp(G',\weight,\activation_{\energyTemplate})$ is optimal for $G'$. %
\end{corollary}

\begin{proof}
	Let the operator for edge-based value iteration for game $(G',\energy(\weight))$ be $\fixpointE'$, and let $\optE$ and $\optE'$ be the edge-optimal value functions for $(G,\energy(\weight))$ and $(G',\energy(\weight))$ respectively.
	As the activation function of the optimal QaSTel is same as the edge-optimal value function,
	it is enough to show that the fixed point computation of $\fixpointE'$ starting from $\optE$ gives the function $\optE'$.

	Let us first shows that $\optE(e) \leq \optE'(e)$ for all $e\in E\setminus\{e_*\}$.
	Let $e\in E\setminus\{e_*\}$ be an arbitrary edge.
	If $\optE'(e) = \infty$, then $\optE(e) \leq \optE'(e)$ trivially holds.
	Suppose $\optE'(e)\neq \infty$, then by definition, for every initial credit $c\geq \optE'(e)$, there exists a $\pz$ strategy $\strat_c$ with $\plays(G',\strat_c)\subseteq \energy_c(\weight)$. As $e_*\in E_0$, the strategy $\strat_c$ is also a well-defined $\pz$ strategy in game graph $G$.
	Hence, for every initial credit $c\geq \optE'(e)$, we have $\plays(G,\strat_c)\subseteq \energy_c(\weight)$, which implies $\optE(e)\geq \optE'(e)$.

	As $\optE'$ is the least fixed point of the monotonic operator $\fixpointE'$ and $\optE\leq\optE'$, by the Knaster-Tarski theorem, the fixed point computation of $\fixpointE'$ starting from $\optE$ gives the function $\optE'$.\qed
\end{proof}

\subsection{Bounding PeSTels}\label{app:boundedPestels}
Given a game graph $\gamegraph = (V,E)$, PeSTels (as defined in \cite{Anandetal_SynthesizingPermissiveWinning_2023}) contain three types of edge templates: (i) \emph{unsafe} edges $\safegroup\subseteq E_0$, (ii)  \emph{co-live} edges $\colivegroup\subseteq E_0$ and (iii) \emph{live} groups $\livegroup\subseteq 2^E_0$. Their combination $\template=(\safegroup,\colivegroup,\livegroup)$ is called a PeSTel, which represents the objective $\plays_\template(G) = \{\play \in V^\omega \mid \forall e\in\safegroup:\ e \not\in\play \text{ and } \forall e\in\colivegroup:\ e \not\in \inf(\play) \text{ and } \forall \livegroupSingleN\in\livegroup:\ \src(\livegroupSingleN)\cap\inf(\play)\neq\emptyset \Rightarrow \livegroupSingleN\cap\inf(\play)\neq\emptyset\}$, where $\src(\livegroupSingleN)$ denotes the source nodes of the edges in $\livegroupSingleN$.

Within this paper, we restrict our attention to PeSTels which only consist of safety 
\& co-live 
templates and call such PeSTels \emph{bounded} as constraint edges 
can never be 
taken unboundedly. %

Given 
a quantitative objective $\spec$ which can be modelled as a parity objective over $G$, one can 
compute a 
bounded PeSTel $\template=(\safegroup,\colivegroup)$ for $(G,\spec)$ by using the algorithm 
\textsc{ParityTemplate} from \cite[Alg.3]{Anandetal_SynthesizingPermissiveWinning_2023} to compute a (full) PeSTel 
$\widetilde{\template}=(\safegroup,\colivegroup,\livegroup)$ first. This PeSTel 
$\widetilde{\template}$ can 
then be \emph{bounded} into a bounded template $\template$ by 
adding all non-live outgoing edges of a 
live-group node, i.e.\ $\{(q,q')\notin H\cup\safegroup|(q,\cdot)\in H\}$ to $\colivegroup$, and then setting $\livegroup:=\emptyset$.

\subsection{Proof of \Cref{prop:extractMixedStrategy}}
\restatePextractMixedStrategy*

\begin{proof}
	By the definition of quantitative conflict-freeness, for every node $ v\in V_0$, there 
	is an edge $ e\in \minActivationEdges(v) $ that is neither colive nor unsafe.
	With this, let us define a positional strategy $\strat$ such that for every node $ v\in 
	V_0 $, $\strat(v)$ is an edge $ e\in \minActivationEdges(v) $ that is neither colive nor 
	unsafe.
	Then, by the definition colive edges, $\strat$ follows the PeSTel $(\safegroup,\colivegroup)$.
	Furthermore, by the proof of \Cref{prop:extractStrategy}, $\strat$ also follows the QaSTel $\energyTemplate$.
	Therefore, $\strat$ follows the mixed template $\Strat$.
\qed\end{proof}

\subsection{Proof of \cref{thm:mixedTemplate}}
\restateMixedTemplateThm*
\begin{proof}
    Let $\Strat = (\safegroup,\colivegroup,\energyTemplate)$.
    The template $\Strat$ is trivially conflict-free due to the while loop in the algorithm. 
    Let after $ i $-th iteration of the while loop, $\spec^i$ be the qualitative objective, $(\safegroup^i , \colivegroup^i) $ be the PeSTel obtained, 
	and $ \energyTemplate^i $ be the QaSTel obtained. Moreover, 
	let $ \win_\spec^i $ and $ \win_\varphi^i $ be the corresponding regions after the $ i $-th iteration.
	For soundness, we need to show that for every $i\geq 0$, $\Strat^i$ is a winning mixed template from $\win^i = \win_\spec^i\cap \win_\varphi^i$ in game $\game$.

	Fix some $i\geq 0$. Let $u\in\win^i$ and let $\strat$ be a strategy following $\Strat^i$, i.e., $ (G,\strat) \vDash_{c}\energyTemplate^i $ for some $ c \geq W \cdot |V| $ and $ (G,\strat) \vDash (\safegroup^i,\colivegroup^i) $.
	By soundness of $\computeTemp$, $\strat$ is winning from $u$ for objective $\spec^i$. As in each iteration, we are just adding additional safety objectives to $\spec$, $\strat$ is also winning from $u$ for objective $\spec$.
	Furthermore, as PeSTels only marks $\pz$'s edges as unsafe or co-live, the conflict edges only belong to $\pz$. Hence, 
	by \cref{prop:correctnessOfHotstarting}, $\strat$ is winning from $u$ for objective $\varphi$.
	So, $\strat$ is a strategy that does not use conflict edges and is winning from $u$ for objective $\varphi$.
	In total, $\strat$ is winning from $u$ for both objectives $\spec$ and $\varphi$, and hence, is winning from $u$ in game $\game$.\qed
\end{proof}

\subsection{Proof of \cref{corr:MPcobuchi}}\label{sec:MPcobuchi}
\begin{corollary}\label{corr:MPcobuchi}
	Given a mean-payoff \cobuchi game $(\gamegraph,\meanpayoff(\weight)\wedge\cobuchi(T))$, if $ (\win, \Strat) =  \mixedTemp(\gamegraph,\weight,\cobuchi(T))$, then $ \Strat$ is a conflict-free winning mixed template from $ \win $. Furthermore, the procedure terminates in time $ \bigO(n^2m+nmW) $, where $ n = \abs{V} $, $ m = \abs{E} $, and $ W$ is the maximum weight in $\weight$.
\end{corollary}

\begin{proof}
	As the soundness directly follows from \cref{thm:mixedTemplate}, we only need to show the complexity.
	We know that each round of while loop uses at most one call to $\cobuchiTemp$ and one call to $\energyTemp$.
	Furthermore, if the winning region does not change in an iteration, then $\cobuchiTemp$ need not be called as the qualitative objective remains the same.
	Then $\cobuchiTemp$ will be called at most $ n $ times, which takes time $ \bigO(n^2m) $ in total.
	Moreover, since we are hot-starting the $ \energyTemp $ algorithm, in total, the time taken by $\energyTemp$ across all iterations is $ \bigO(nmW) $. Hence, the algorithm terminates in time $ 
	\bigO(n^2m + nmW) $.\qed
\end{proof}
\section{Plots from the experiments}

\begin{figure}[t]\label{plotsInAppendix}
	\centering
	\begin{subfigure}[b]{0.95\textwidth}
		\includegraphics[width=\textwidth]{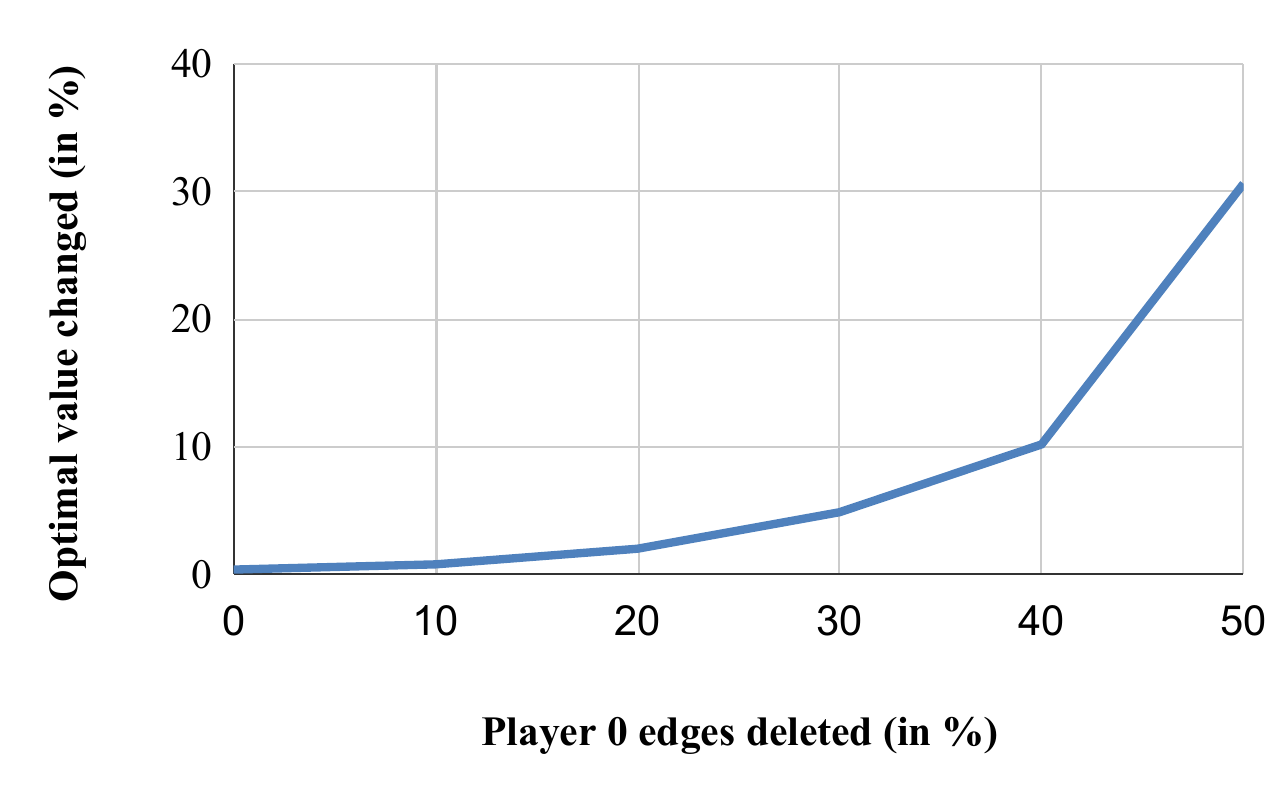}
		\caption{Fault tolerance}\label{plotsApp_faultTolerance}
	\end{subfigure}
	
	\vspace{1em}
	
	\begin{subfigure}[b]{0.95\textwidth}
		\includegraphics[width=\textwidth]{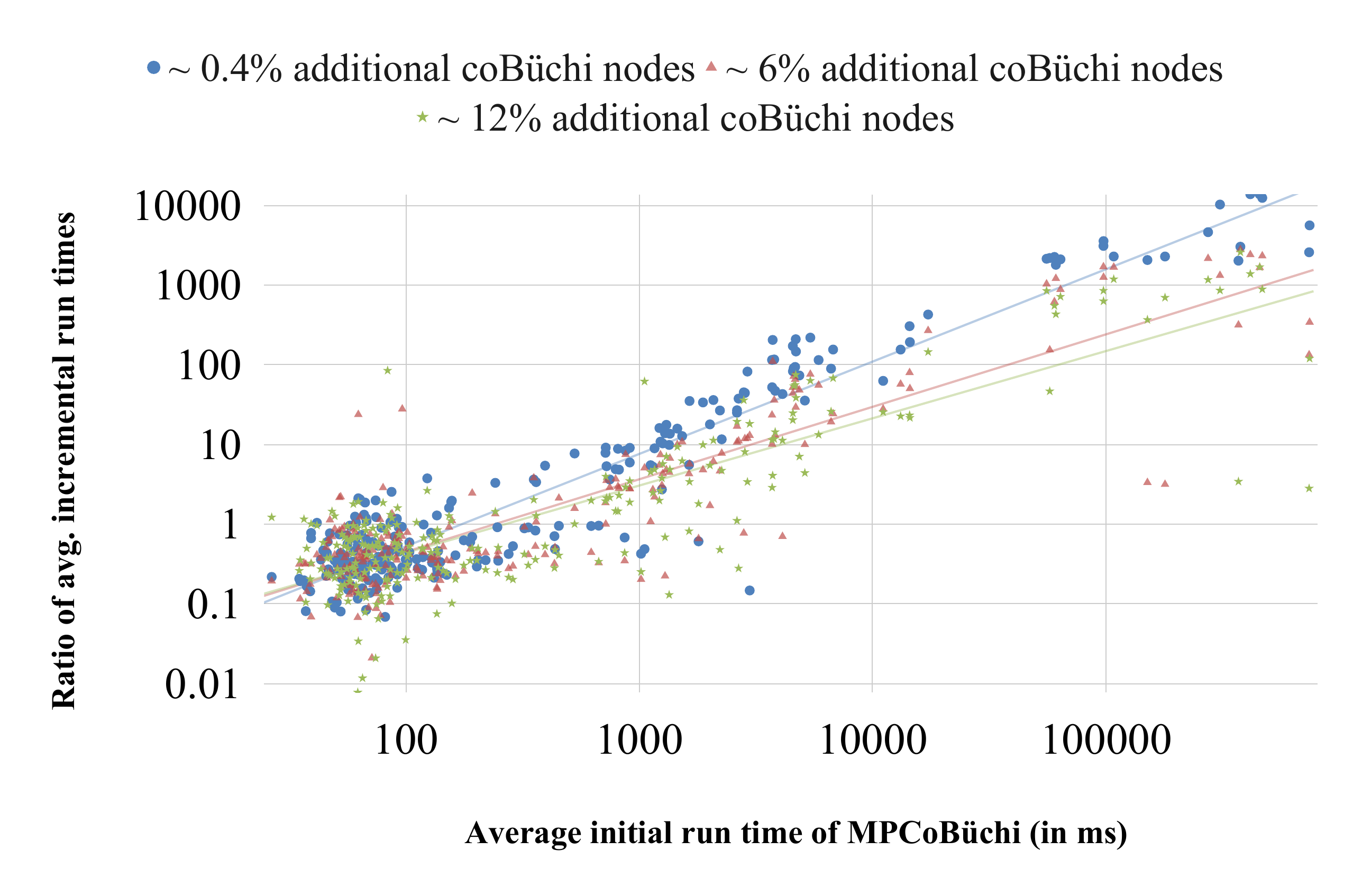}
		\caption{Incremental synthesis}\label{plotsApp_incremental}
	\end{subfigure}
	
	\vspace{1em}
	
	\begin{subfigure}[b]{0.95\textwidth}
		\includegraphics[width=\textwidth]{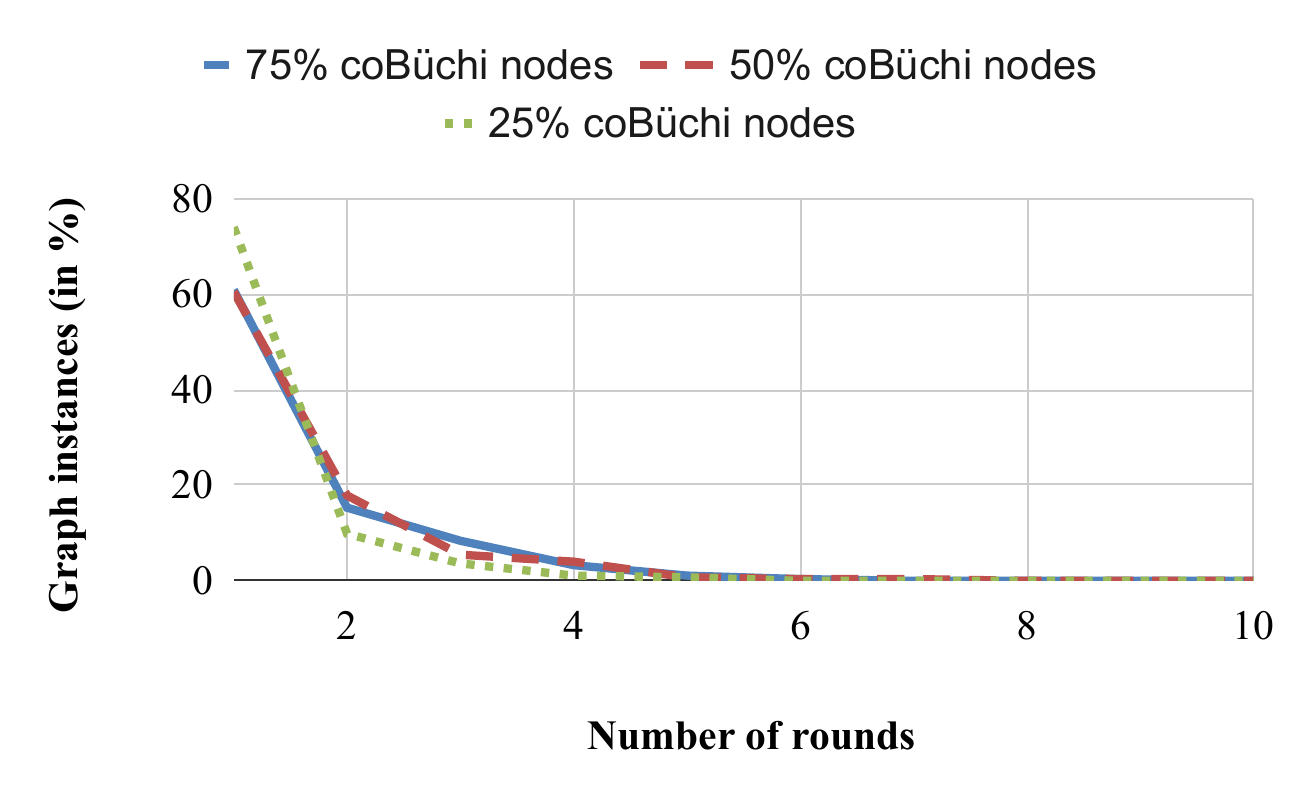}
		\caption{Rounds of conflict resolution}\label{plotsApp_allMPCoBuechi}
	\end{subfigure}
\end{figure}

\begin{figure}[t]
	\ContinuedFloat
	\centering
	
	\begin{subfigure}[b]{0.95\textwidth}
		\includegraphics[width=\textwidth]{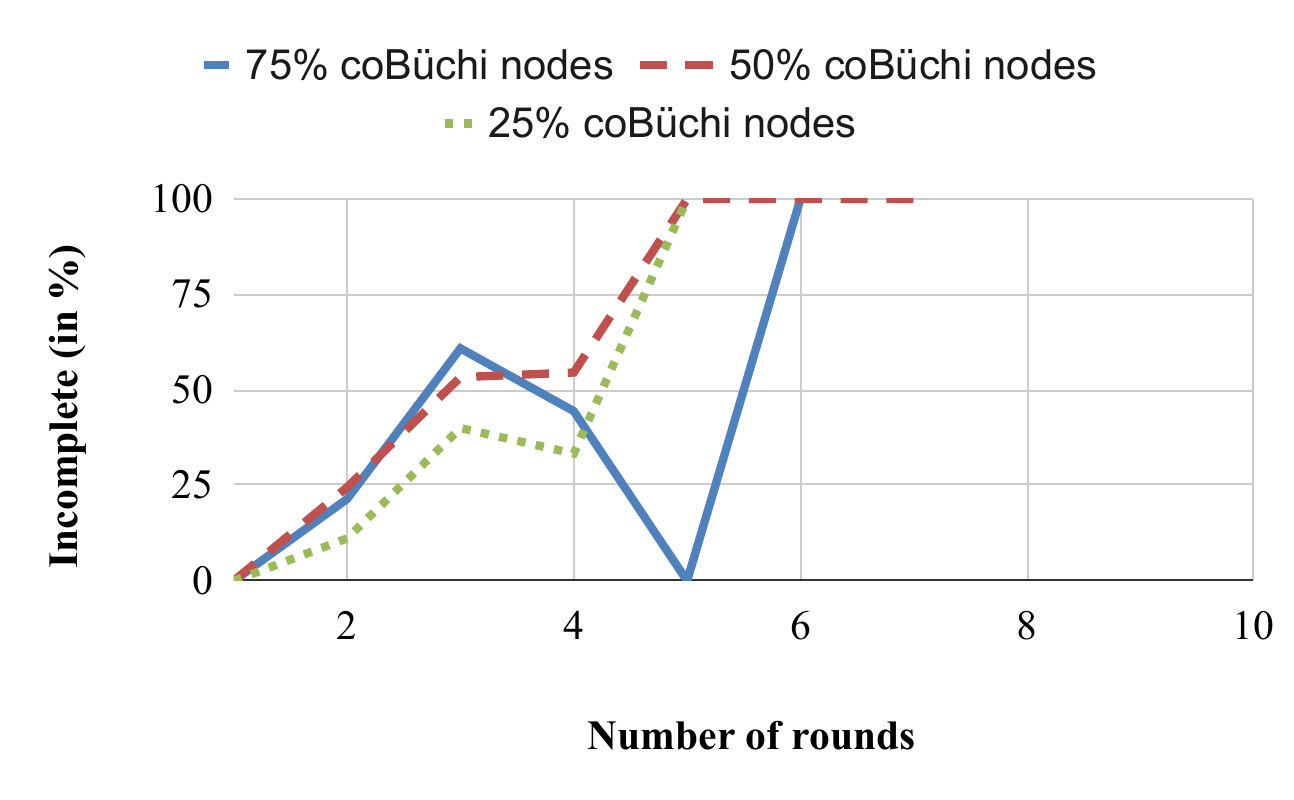}
		\caption{Incompleteness of \tool}\label{plotsApp_incomplete}
	\end{subfigure}
	
	\caption{Plots summarizing the experimental evaluations.}
\end{figure}